\documentclass[11pt]{article}
\usepackage{enumerate}
\usepackage{optidef}
\usepackage[utf8]{inputenc}
\usepackage[english]{babel}

\usepackage[letterpaper,margin = 1in]{geometry}

\usepackage{graphicx}
\usepackage{subcaption}
\usepackage{xcolor}
\usepackage{tikz}
\usepackage{booktabs}
\usepackage{fancyhdr}
\usepackage{microtype}
\usepackage{indentfirst}
\usepackage{hyperref}
\usepackage{cleveref}
\usepackage{listings}

\usepackage{amsmath}
\usepackage{amsthm}
\usepackage{amssymb}
\usepackage{amsfonts}
\usepackage{mathtools, thmtools, thm-restate}
\usepackage{algorithm}
\usepackage{algpseudocode}
\usepackage{float}

\newtheorem{thm}{Theorem}

\newtheorem{defn}[thm]{Definition}

\newtheorem{lemma}[thm]{Lemma}
\newtheorem{corollary}[thm]{Corollary}

\lstdefinestyle{defaultstyle}{
	backgroundcolor = \color{lightgray!40!white},
	basicstyle = \ttfamily\footnotesize,
	keywordstyle = \color{purple},
	numbers = left,
	numberstyle = \tiny\color{gray}
}
\title{Controlling tail risk in two-slope ski rental\thanks{Supported in part by NSF award 2228995.}}
\author{Qiming Cui\\Johns Hopkins University\\Baltimore, MD USA\\ \texttt{qcui6@jhu.edu} \and Michael Dinitz\\Johns Hopkins University\\Baltimore, MD USA\\ \texttt{mdinitz@cs.jhu.edu}}
\date{}

\begin{document}
	
\maketitle

\begin{abstract}
    We study the optimal solution to a general two slope ski rental problem with a tail risk, i.e., the chance of the competitive ratio exceeding a specific value $\gamma$ is bounded by a constant $\delta$. This extends the recent study of tail bounds for ski rental initiated by [Dinitz, Im, Lavastida, Moseley, Vassilvitskii SODA 2024] to the two-slope version of ski rental defined by [Lotker, Patt-Shamir, Rawitz IPL 2008].  In this version, even after ``buying'' we must still pay a rental cost at each time step (unlike the traditional setting), but the rental cost is lower after buying than it is beforehand.  This models many real-world ``rent-or-buy'' scenarios where we can invest in some sort of improvement for a one-time cost, which then decreases (but does not eliminate) the per-time cost.  
    
    Despite this being perhaps the simplest extension of the classical ski-rental problem, we find that after adding bounds on the tail risk the optimal solution has fundamentally different structure from the classical setting with tail bounds, which already had significantly different structure compared to settings without tail bounds.  For example, in our setting there is a possibility that we never buy in an optimal solution (which can also occur without tail bounds), but more strangely (and unlike the case without tail bounds or the classical case with tail bounds) we also show that the optimal solution might need to have nontrivial probabilities of buying even at finite points beyond the time corresponding to the buying cost.  Moreover, in many regimes there does not exist a unique optimal solution.  As our first contribution, we develop a series of structure theorems to characterize some features of optimal solutions.  
    
    The complex structure of optimal solutions makes it more difficult to develop an algorithm to compute such a solution.  For example, it is not a priori clear whether the support in an optimal solution is bounded (and indeed there are optimal solutions where the support is unbounded).  Nevertheless, as our second contribution we utilize features of our structure theorems to design two algorithms, one based on a greedy algorithm combined with binary search that is fast but yields solutions that are suboptimal (although they are arbitrarily close to optimal), and a slower algorithm based on linear programming which computes exact optimal solutions.  
\end{abstract}

\section{Introduction}
\label{sec:introduction}
The ski rental problem is one of the most fundamental problems in online algorithms, modeling the fundamental ``rent-or-buy'' decision making task that arises constantly in real decision-making settings.  In this problem there is an unknown number of ski days, and before each day we can choose to either rent skis for the day for \$1 or buy skis for \$b.  If we have perfect knowledge then this problem is trivial: we rent each day if the number of ski days is at most $b$, and otherwise we buy skis at the beginning of day $1$.  The difficulty arises from \emph{not} knowing the number of ski days in advance, and instead being forced to make our rent/buy decision \emph{online}.  Ski rental has been the object of intensive study since its introduction~\cite{Karlin1986CompetitiveSC}, and has long been known to admit a deterministic $2$-competitive algorithm (an algorithm that never pays more than twice the optimal cost in hindsight)~\cite{Karlin1986CompetitiveSC} and a randomized algorithm that is $\frac{e}{e-1} \approx 1.58$-competitive in expectation~\cite{Karlin1990CompetitiveRA}.  Note that a randomized algorithm is simply a probability distribution over days; we will call this the \emph{purchase distribution} of the algorithm.

Recently Dinitz, Im, Lavastida, Moseley, and Vassilvitskii~\cite{Dinitz2024ControllingTR} initiated the study of \emph{tail bounds} for the classical randomized algorithm.  They showed that the classical algorithm does not exhibit concentration, and has a reasonably large probability of actually performing worse than the classical deterministic algorithm.  To rectify this, they considered algorithms that minimize the expected competitive ratio \emph{subject} to satisfying explicit tail bounds which they called $(\gamma, \delta)$-constraints: they required that their algorithm have probability at most $\delta$ of having competitive ratio worse than $\gamma$.  They were able to construct an (efficient) algorithm to compute the optimal purchase distribution for any collection of such constraints, but perhaps more surprisingly, they showed that optimal purchase distributions satisfying even a single $(\gamma, \delta)$-constraint exhibited dramatically different structure from the classical randomized algorithm (without tail bounds).  For example, while the classical algorithm is a simple exponentially increasing distribution in $[1,b]$, with a single $(2,\delta)$-constraint the optimal purchase distribution becomes non-monotonic, can have arbitrarily many discontinuities, and even where there are not discontinuities the purchase distribution can grow arbitrarily quickly (i.e., have arbitrarily large derivative).  

Ski rental is a fundamental problem because it abstracts a class of problems in which we can either pay a (small) repeating cost (the rental fee) or can pay a single one-time cost (the buying fee).  However, many problems exhibit a different but extremely related structure: one can again pay a repeating cost, or one can pay a single one-time cost which, rather than eliminating the rental fee, instead \emph{reduces it}.  This models, for example, any manufacturing setting where one can buy a machine, piece of infrastructure, or invest in a piece of technology that makes some process more efficient: we still need to pay every day to run the better machine, but since it is more efficient we pay less per output.  It is again trivial to determine whether to buy such an improved machine if we know the number of days that the factory will run, and the difficulty comes from \emph{not knowing} this time horizon.  

In order to model this, Lotker, Patt-Shamir, and Rawitz~\cite{Lotker2008SkiRW} introduced the \emph{two-option generalized ski rental problem}.  As usual, we first switch to a continuous setting where $b=1$ and time is continuous starting at $0$.  Then the classical ski rental setting can be thought of as having a choice of two states, each of which incurs some cost per time to be in: state 1, which costs \$1 per time to be in, and state 2, which costs \$0 per time to be in.  Moreover, it costs \$1 to switch from state 1 to state 2.  In other words, the two states can be modeled as \emph{lines}: if we spend $t$ time in state $1$ then we pay $1\cdot t + 0$, while if we spend $t$ time in state $2$ we pay $0\cdot t + 1$.  Lotker et al.\cite{Lotker2008SkiRW} generalized this to allow the second state to \emph{also} have some cost $0 \leq a \leq 1$ per time and switch a switching cost of $1-a$, corresponding to the line $a\cdot t + (1-a)$.  Note that any algorithm for this problem is similarly a purchase distribution that just specifies a distribution over when to switch from state 1 to state 2.  Their main result was an optimal randomized algorithm that is $\frac{e}{e-1+a}$-competitive in expectation for this problem.  Moreover, the optimal purchase distribution corresponding to this algorithm is remarkably similar to the randomized algorithm for the classical ski rental problem: it is essentially a scaled down version of the same exponential distribution in $[0,1]$, but with some extra probability mass on the ``never switch'' option (corresponding to placing probability mass at time $\infty$).  

By combining these two lines of work we get a number of interesting and important questions.  For the two-state ski rental problem in the presence of tail bounds, can we still design an algorithm to find the optimal purchase distribution, as in Dinitz et al.~\cite{Dinitz2024ControllingTR}?  Since the structure of the two-state solution~\cite{Lotker2008SkiRW} and the classical (``pure buy'') solution~\cite{Karlin1990CompetitiveRA} are extremely similar (just adding some probability mass at time $\infty$), when we add tail constraints do we similarly get the same basic structure as in~\cite{Dinitz2024ControllingTR}?

\subsection{Our Contributions}
In this paper we show that the structure of optimal purchase distributions in the presence of even a single tail bound is even more complicated than in Dinitz et al.~\cite{Dinitz2024ControllingTR}, but that we can still design efficient algorithms to compute this distribution (even for a collection of tail bounds).  For example, we prove that optimal purchase distributions can have the following structures, in contrast to the pure buy setting  (we follow~\cite{Lotker2008SkiRW} and refer to the classical setting as the ``pure buy'' setting).
\begin{itemize}
\item There may not be a unique purchase distribution, despite the fact that there is always a unique optimal purchase distribution for either the pure buy case with a single tail bound~\cite{Dinitz2024ControllingTR} or the two-state setting without tail bounds~\cite{Lotker2008SkiRW}.  
\item For some parameter settings, \emph{every} optimal solution might include probability mass at finite points larger than $1$ (or in the discrete case, larger than $b$).
Again, this is in contrast to the pure buy case even with a collection of tail bounds~\cite{Dinitz2024ControllingTR} (where the optimal purchase distribution only has support in $[0,1]$) and the two-state case without tail bounds~\cite{Lotker2008SkiRW} (where the optimal solution only has support in $[0,1] \cup \{\infty\}$).
\end{itemize}

Note that these features make it more difficult to compute an optimal purchase distribution.  In particular, since we not have an \emph{a priori} bound on the support, it is not clear that one can even \emph{represent} an optimal solution efficiently, much less compute it.  Nevertheless, we design algorithms to compute optimal purchase distributions (in the discrete setting).  

\paragraph{Definitions.} 
In order to more formally present our results and techniques, we first need to introduce some notation.  
In this paper, we use time $t$ and point $t$ interchangeably. A \emph{purchase distribution} $f$ is a probability distribution over $\mathbb{R}_{\geq 0} \cup \{\infty\}$, where $f_x$ represents the probability that we buy (move from state 1 to state 2) at time $x$ (we will use $f$ to mean both the distribution and its probability density function interchangeably).  Note that since our actions are probabilistic, our competitive ratio is now a random variable.  We will let $\alpha_f(x)$ denote the expected competitive ratio when the true stopping time is $x$ and we choose our buying time from $f$.  So our goal is to compute the purchase distribution which minimizes $\alpha_f = \sup_x \alpha_f(x)$
subject to also obeying the given tail bound(s).  

Given a $(\gamma, \delta)$-tail bound, the \emph{bad interval} of $x$ (denoted by $I_{\gamma}(x)$) is defined to be the set of all points $t$ such that if we buy at time $t$ and the adversary chooses time $x$ to stop, then the competitive ratio is larger than~$\gamma$. This is essentially the same definition of bad intervals as in~\cite{Dinitz2024ControllingTR}, but it turns out that their structure is quite different: most notably, in the pure buy setting it is not hard to argue that the bad interval for $x$ is contained in $[0,x]$, but in the two-state setting this is no longer true: the bad interval for $x$ can contain points in $(x, \infty)$.
We will refer to an interval $I$ of $\mathbb{R}_{\geq 0}$ as a \emph{zero interval} for $f$ if there is no probability mass in $I$, i.e., if $\int_{x \in I} f_x\ dx = 0$. When the right boundary is $\infty$, we call it a zero \emph{suffix}. 

Note that any purchase distribution with expected competitive ratio $\beta$ that satisfies a $(\gamma, \delta)$-tail constraint must satisfy two properties at every point $x$. First, it must be within the competitive ratio constraint: $\alpha_f(x) \leq \beta$.  Second, in order to obey the tail constraint there must be at most $\delta$ probability mass in the bad interval for $x$: $\int_{t \in I_{\gamma}(x)} f_t\ dt \leq \delta$.  We refer to the first of these as the competitive ratio constraint (or CR-constraint) and the second as the mass constraint. We note that the continuous and discrete settings are essentially the same for our purposes, so we move between them depending on what is convenient. The exception is for our algorithms, which only work in the discrete setting.

\subsubsection{Structure and Greedy}

In order to understand the structure of the optimal purchase distribution in the presence of tail bounds in the pure buy setting, Dinitz et al.~\cite{Dinitz2024ControllingTR} proved a structural theorem about the optimal distribution: at every point $x$, either the CR-constraint with $\beta = OPT$ or the mass constraint must be tight (or both).  And since (as discussed) the bad interval for every point $x$ is contained in $[0,x]$, this led them to an obvious greedy algorithm: moving from smaller to larger values of $x$, increase $f_x$ until one of the two constraints becomes tight.  While this is not actually an implementable algorithm as stated (since we do not know $OPT$), it enabled them to later design an actual algorithm as well as allowing them to reason about the optimal solution by reasoning about the behavior of the greedy algorithm.  

Unfortunately, it turns out that this theorem is \emph{not} true in the general two-state setting: unlike the pure buy setting, there can indeed be optimal purchase distributions $f$ which have points $x$ where $f_x > 0$ but neither of the constraints is tight at $x$.  Moreover, as discussed, in the two-state setting the bad interval for $x$ can actually contain points \emph{after} $x$, forcing us to add a third requirement when increasing $f_x$ to ensure that we don't violate the mass constraint for some $x' < x$.  Nevertheless, despite these difficulties, we prove the following theorem.

\begin{restatable}{thm}{OurContributionTheoremONEGreedyGivesOptimal}
    \label{our contribution:thm 1, greedy gives optimal}
    Given the true value of OPT, the greedy algorithm gives an optimal solution.
\end{restatable}

This theorem allows us to explore the structure of at least some optimal solutions, in particular those that are closely related to the greedy solution (the purchase distribution output by the greedy algorithm). 
Moreover, this allows us to prove that the structure of optimal solutions exhibits a phase transition at a specific value of $\delta$.  For large values of $\delta$, the structure is actually quite similar to prior work:

\begin{restatable}[Structure Theorem, large $\delta$]{thm}{OurContributionTWOLargeDelta}
    \label{our contribution: thm 2, general, large delta}
    If $\frac{OPT-1}{\frac{1}{a}-1} \leq \delta \leq 1$, then there is a unique optimal solution $f$, and it has the following structure: every point in $[0,1]$ has either the CR constraint or the mass constraint tight, $f_x = 0$ for all $x > 1$, and $f_{\infty} > 0$.
\end{restatable}

In particular, in this range of $\delta$, the difference between the two-state with tail constraints setting and pure buy with tail constraints setting~\cite{Dinitz2024ControllingTR} is just that extra mass is added at $\infty$.  This is completely analogous to the difference between the two-state setting without tail constraints~\cite{Lotker2008SkiRW} and the pure buy setting without tail constraints~\cite{Karlin1990CompetitiveRA}, and shows that the structure of the optimal distribution is almost precisely as ``complex'' as in the pure buy setting.  

When $\delta$ is below $\frac{OPT-1}{\frac{1}{a} - 1}$, the two-state setting with tail bounds behaves fundamentally differently from the pure buy setting with tail bounds~\cite{Dinitz2024ControllingTR}.
However, we show that in this small $\delta$ regime the greedy solution exhibits particular structure, which will enable us to actually design an algorithm in later sections.  Most notably, we show that while the greedy algorithm places some probability mass at finite points beyond $1$,
there is actually a point beyond which the greedy algorithm will not place any mass (other than at $\infty$).  Specifically, we give the structure theorem (Theorem~\ref{our contribution:thm 3, general, small delta}) for general optimal purchase distributions in this small $\delta$ regime.

\begin{restatable}[Structure Theorem, small $\delta$]{thm}{OurContributionThreeSmallDelta}
    \label{our contribution:thm 3, general, small delta}
    If $0 \leq \delta < \frac{OPT-1}{\frac{1}{a}-1}$, then:
    \begin{itemize}
        \item There is not a unique optimal purchase distribution.
        \item There exists a finite point $L_b = \frac{\gamma-1}{1-a\gamma} \geq 1$, such that there must be exactly $\delta$ mass in the suffix $(L_b,\infty]$ in any optimal purchase distribution.
    \end{itemize}
\end{restatable}

We can now give a structure theorem specifically for the greedy solution.  Thanks to the specificity of the greedy algorithm, we end up with a thorough understanding of the resulting solution.

\begin{restatable}[Structure Theorem for the greedy solution, small $\delta$]{thm}{OurContributionFOURGreedyStructure}
    \label{our contribution:thm 4, greedy solution, small delta}
    If $0 \leq \delta < \frac{OPT-1}{\frac{1}{a}-1}$, then, there exists a finite point $p \leq L_b$, such that the greedy solution has the following structure: 
    \begin{itemize}
        \item In $[0,1]$, either the $CR$ or mass constraint is tight at every point.
        \item In $(1,p]$, the $CR$ constraint is tight at every point and the solution behaves as an exponential function in $(1,p]$. Specifically, for small time step $\tau$, $f_{1+k\tau} = (1+\tau)^{k-1}f_{1+\tau}$ where $1 \leq k \leq \lfloor \frac{p-1}{\tau} \rfloor$.
        \item In $(p,L_b]$, the mass constraint is tight at every point.
        \item After $L_b$, the solution has either $CR$ or mass constraint tight at each point until there is a total of $\delta$ mass. Moreover, it has finite support.
    \end{itemize}
\end{restatable}

One can see that the full structure of the greedy solution is quite different from the pure buy setting~\cite{Dinitz2024ControllingTR} where the support is always $[0,1]$ even with tail bounds.  In fact, the structures are even more different than they might appear.  In the pure buy setting there are also regimes where the mass constraint must be tight (as in the $(p, L_b)$ interval for us), which was shown by~\cite{Dinitz2024ControllingTR} to result in a series of exponential functions with increasing base (when combined with the fact that the bad intervals for such points are decreasing in length).  In our setting, since the bad intervals for points in $(p, L_b)$ are actually \emph{increasing} in length, this results in exponential functions with \emph{decreasing} base. 
And even more strangely, due to the discrete time setting we also end up with frequent interruptions by zero intervals (corresponding to points where the left endpoint of the bad interval does not actually move thanks to the discretization).  To see an example of this, see Figure~\ref{fig:exponential_interrupted_by_zero_intervals}, where time has been discretized finely enough to result in what appear to be overlapping zeros and exponentials (although they are in reality alternating, not overlapping).

Moreover, while the greedy solution always has finite support, it is not clear whether this support can be written as a closed analytical form of the input values. So, we show that there is a closely related distribution (which is a simple modification of the greedy solution) which always has support in $[0,L_b] \cup \{\infty\}$, giving a ``fixed" finite support solution. Together with the fact that there is a unique optimal solution with support $[0,1] \cup \{\infty\}$ in the large $\delta$ regime, we have the following structural theorem for the fixed finite support solution for all $\delta$.

\begin{restatable}[Structure Theorem for the fixed finite support solution, all $\delta$]{thm}{OurContributionFiveAllDelta}
    \label{our contribution:thm 5, closely related greedy, all delta}
    For every $\delta$, there is an optimal purchase distribution $f$ which has support in $[0,L_b] \cup \{\infty\}$. In other words, $f_x = 0$ for all $x \in (L_b,\infty)$.
\end{restatable}

\subsubsection{Algorithms} 
After exploring the structure of optimal solutions, we give two different algorithms to solve the problem: a binary search greedy algorithm and an LP-based algorithm. The binary search greedy algorithm takes in an accuracy parameter $\epsilon>0$ and a guess $T$ of $OPT$, and applies the aforementioned greedy algorithm using $T$ as $OPT$ to construct the $\{f_t\}$ values.  However, since $T$ is not necessarily equal to $OPT$, these values might not actually form a valid distribution.  We prove that $\sum_t f_t > 1$ if and only if $T > OPT$.  This allows us to do binary search for $OPT$, enabling us to return a feasible solution (a purchase distribution that obeys the $(\gamma, \delta)$-tail constraint) with expected competitive ratio at most $OPT+\epsilon$.

\begin{restatable}[Binary Search Greedy]{thm}{OurContributionSixBinarySearch}
    \label{our contribution:thm 6, binary search}
    There is a binary search greedy algorithm that gives feasible solution with expected competitive ratio at most $OPT+\epsilon$ in $O(\log(1/\epsilon))$ iterations.
\end{restatable}

While this algorithm is fast and effective, it does not actually return an optimal solution (just a solution that is arbitrarily close to optimal).  To remedy this, we take advantage of the fact that there is a related optimal solution which has \emph{fixed} finite support, and in particular has support in $[0,L_b] \cup \{\infty\}$ (Theorem~\ref{our contribution:thm 5, closely related greedy, all delta}).
Thanks to this theorem, we can write a linear program in which the variables are the purchase distribution and the constraints are the CR and mass constraint (even without knowing $OPT$).  By solving this LP we get a polynomial-time algorithm to find an exactly optimal purchase distribution, at the cost of larger running time to solve the LP.

\begin{restatable}[LP Algorithm]{thm}{OurContributionSevenLP}
    \label{our contribution:thm 7, LP}
    There is a linear program with $\Theta(L_b/\tau)$ variables and $\Theta(L_b/\tau)$ constraints where any optimal solution corresponds to an optimal purchase distribution.  
\end{restatable}

\paragraph{Experiments.} Since we have efficient algorithms to compute optimal and near-optimal purchase distributions, we have the ability to experimentally explore some of the parameter space to illustrate the surprising structures of the optimal purchase distributions (as discussed in our structure theorems). 
We exhibit the optimal purchase distributions and the figures that illustrates CR and mass constraints in both large and small $\delta$ regimes. Specifically, in the small $\delta$ regime, we observe that the optimal purchase distributions generated by the above two algorithms not only have distinct structures after $L_b$ as stated in the structure theorems, but also reveal differing structures before~$L_b$: the distribution given by the LP algorithm can be non-greedy in $[1,L_b]$, thus the optimal purchase distribution is not unique even before $L_b$.  Also, as already discussed, we present a figure showing different behaviors in $(p,L_b]$, indicating that to make mass tight, the optimal purchase distribution consists of alternating exponentials and zero intervals.
These results and some discussion can be found in Section~\ref{sec:Simulations}.

\subsection{Related Work}

The classical ski rental problem, also known as the rent or buy problem, is the most basic problem considered in online algorithms. The literature on it and its variants is deep and vast, so we do not attempt to survey all such work, but rather only that which is most directly related to our setting.  Karlin et al.~\cite{Karlin1986CompetitiveSC} first analyzed the 2-competitive deterministic strategy and then gave the $e/(e-1)$-competitive randomized algorithm in \cite{Karlin1990CompetitiveRA} in the classical pure buy setting. Separately, Karlin et al. \cite{Karlin2003DynamicTA} applied the randomized algorithm to TCP acknowledgement and some other problems. We note that there are multiple variants of the classical ski rental problem. Lotker, Patt-Shamir, and Rawitz \cite{Lotker2008SkiRW} indicated the optimal $e/(e-1+a)$-competitive randomized algorithm in the two states setting. They also studied the multislope ski rental \cite{Lotker2008RentLO} and gave an optimal randomized strategy with an $e/(e-1)$ competitive ratio in the additive model and with an $e$ competitive ratio in the non-additive model. Later, Levi and Patt-Shamir \cite{Levi2013NonadditiveTS} improved the competitive ratio for the non-additive model with two slopes. Fujiwara et al. \cite{Fujiwara2011OnTB} explored the upper and lower bounds for the competitive ratio in the $k$ slope setting for small values $k$. The analysis for different variants of the ski rental problem can be found in \cite{Khanafer2013TheCS}\cite{Khanafer2015ToRO}\cite{Madry2011TheSS}.

As discussed, Dinitz et al.~\cite{Dinitz2024ControllingTR} is the first to consider controlling the tail risk for the classical ski rental problem. They gave the optimal randomized strategy and analyzed the structure of the optimal solution, which is significantly different from the no tail constraint case.
Subsequently, Christianson et al.~\cite{christiansonCOLT} introduced a related type of tail bound called the \emph{conditional} value-at-risk (as opposed to the \emph{absolute} tail bounds of~\cite{Dinitz2024ControllingTR}). 
This type of tail bound turns out to be more amenable to analysis, and they were able to derive a number of interesting bounds and algorithms.

\subsection{Paper outline}
We begin in Section~\ref{sec:Preliminaries} by formally defining the problem and important concepts. We expand upon this in Appendix~\ref{appendix:zero suffix and zero interval}, where we prove important properties of zero intervals and zero suffixes that will be useful for our main proofs.  We then explore the structure of the optimal purchase distributions in Section~\ref{sec:structure of optimal solutions}, proving that the greedy algorithm gives an optimal solution (Theorem~\ref{our contribution:thm 1, greedy gives optimal}) in Section~\ref{subsec: optimal solution by greedy algorithm} and then argue about the structure of optimal solutions in the two different regimes of $\delta$ (Theorems~\ref{our contribution: thm 2, general, large delta} to~\ref{our contribution:thm 5, closely related greedy, all delta}) in Section~\ref{subsec: general structure of optimal solutions}. In Appendix~\ref{sec:algorithms} we give our two algorithms for computing optimal purchase distributions (Theorem~\ref{our contribution:thm 6, binary search} and~\ref{our contribution:thm 7, LP}) and discuss the trade-off between the algorithms. Finally, in Appendix~\ref{sec:Simulations} we show different structures by exploring different parameters experimentally.

\section{Preliminaries}
\label{sec:Preliminaries}

We restate our model here. We consider the ski rental problem with tail constraints in the two states setting (defined in Section~\ref{sec:introduction}). State 1, where we call it the rental state, is to pay one unit of rental cost per time unit. To leave state 1 and reach state 2, one needs to pay a one-time cost $1-a$ and then the rental rate drops to $a$, for some $0 \leq a < 1$. If we see the geometric interpretation of the model, it consists of two intersecting lines with the intersection point $(1,1)$. If the adversary chooses the stopping time $t \leq 1$, then the optimal strategy is to stay in state 1 from the beginning, with an optimal cost $t$. If the adversary chooses $t > 1$ to stop, then the optimal strategy is to switch to state 2 at the beginning, with an optimal cost $1-a+at$. After one moment of thinking, the best deterministic algorithm is $(2-a)$-competitive. Our task is to determine when to switch to state 2 by giving the purchase distribution, at the same time not violate the tail constraints $(\gamma,\delta)$ where generally we have $\gamma \geq 2-a$ and $0 \leq \delta \leq 1$.

\subsection{Bad Interval}
Let $\{f_t\}_{t \geq 0}$ be the purchase distribution function, where $f_{\infty}$ is the probability that we never buy. Note that $x,t \in \mathbb{R^{+}} \cup \{\infty\}$. Let $\alpha(t,x)$ be the competitive ratio of the algorithm when adversary chooses for $x$ to be the last time of skiing and we buy at the time $t$, and $a$ is the slope of the second option. Then we have
\[\alpha(t,x) = \begin{cases} 
      \frac{t+1-a+a(x-t)}{x},  & t \leq x \ \text{and} \ 0 \leq x \leq 1 \\
      1,  & t > x \ \text{and} \ 0 \leq x \leq 1 \\
      \frac{t+1-a+a(x-t)}{1-a+ax},  & t \leq x \ \text{and} \ x>1 \\
      \frac{x}{1-a+ax},  & t > x \ \text{and} \ x>1.
   \end{cases}\]
Together with the tail constraint $(\gamma,\delta)$ where $\gamma \geq 2-a$ and $0 \leq \delta \leq 1$, we define the bad interval of a point $x$.

\begin{defn}
    \label{def: def of bad interval}
    The bad interval $I_{\gamma}(x)$ consists of all time $t$ such that if we buy at time $t$ and adversary chooses time $x$, then the competitive ratio is larger than $\gamma$. In other words,
    \[I_{\gamma}(x)=\{t:\alpha(t,x) > \gamma\}=\begin{cases} 
      \{t \leq x: \frac{t+1-a+a(x-t)}{x}> \gamma\},  & \text{for} \ 0 \leq x \leq 1 \\
      \{t \leq x: \frac{t+1-a+a(x-t)}{1-a+ax} > \gamma\} \cup \{t>x:\frac{x}{1-a+ax}>\gamma\},  & \text{for} \ x>1
   \end{cases}\]
\end{defn}

For every point $x$, we want to simplify and explicitly write down the formula of its bad interval. We note that bad intervals may have different forms. For sake of completeness, we denote two bound points $I_3 \coloneqq \frac{(\gamma-1)(1-a)}{1-a\gamma}$ and $L_b=\frac{\gamma-1}{1-a\gamma}$, where we show in Appendix~\ref{appendix:bad intervals in appendix} that $I_3$ is the last point whose bad interval is empty and $L_b$ is the largest point not contained in any bad interval that stretches out to $\infty$. We defer the calculation and more discussions (bounds of different bad intervals) to Appendix~\ref{appendix:bad intervals in appendix}.

\subsection{Feasible solution and Optimal solution}

\begin{defn}
    \label{def: def of feasible distribution}
    A distribution $f$ is feasible if it satisfies the constraint with respect to $(\gamma, \delta)$ where $\gamma \geq 2-a$ and $0 \leq \delta \leq 1$. We call $f$ a feasible solution if the distribution is feasible.
\end{defn}

The following lemma gives a way to check the feasibility of a distribution.

\begin{lemma}
    \label{lem: f feasible formula}
    $f$ is feasible if and only if $\int_{t \in I_{\gamma(x)}} f_t \, dt \leq \delta$ for all $x \in \mathbb{R^{+}} \cup \{\infty\}$.
\end{lemma}
\begin{proof}
    Follows directly from the definition.
\end{proof}

Now, we can discuss if a feasible solution $f$ is an optimal solution. Given a feasible solution~$f$, let $\alpha_f(x)$ be the expectation of $\alpha(t,x)$ where $t$ is drawn from $f$.

If $0 \leq x \leq 1$,
\[\alpha_f(x) = \int_{t=0}^{\infty} \alpha(t,x)f_t \, dt = \int_{t=0}^{x} \frac{t+1-a+a(x-t)}{x}f_t \,dt + \left(\int_{t=x}^{\infty} f_t \,dt + f_\infty\right).\]

If $1 \leq x$,
\[\alpha_f(x) = \int_{t=0}^{x} \frac{t+1-a+a(x-t)}{1-a+ax}f_t \,dt + \left(\int_{t=x}^{\infty} \frac{x}{1-a+ax}f_t \,dt + \frac{x}{1-a+ax}f_\infty\right).\]

It is not hard to see that $\alpha_f(x)$ is continuous.

We say that a feasible solution is optimal if it gives the best worst-case competitive ratio. Let $f^{*}$ be an optimal solution. By Lemma~\ref{lem: f feasible formula}, $f^{*}$ is the distribution $f$ that
\begin{mini*}|s|
{}{max_{x} \ \alpha_f(x)}
{}{}
\addConstraint{\int_{t \in I_{\gamma(x)}}f_t \, dt \leq \delta.}
{}
\end{mini*}

For some point $x$, if the above inequality constraint is an equality constraint, then it is called tight. One can forward to the next subsection for the formal definition of tightness.

\subsection{Tightness of the mass constraint and the Competitive Ratio constraint}

Let $f$ be a feasible solution. Denote $OPT \coloneqq \text{max}_x \alpha_{f}(x)$. Note that we have a closed analytical form of $OPT$ if we have $OPT = \alpha_{f}(\infty)$. In this case, $OPT$ can be computed as
\[OPT = \alpha_{f}(\infty)=1 \cdot (1-f_{\infty}) + \frac{1}{a} \cdot f_{\infty} = 1 + \left(\frac{1}{a}-1\right)f_{\infty}.\]

A useful observation is that both the competitive ratio and mass at time $\infty$ are functions only of the mass at $\infty$. Now we define both the tightness of the mass constraint and the competitive ratio ($CR$) constraint.
\begin{defn}
    \label{def: Tight mass constraint}
    We say that the mass constraint is tight at $x$ if $\int_{t \in I_{\gamma(x)}}f_t \, dt = \delta$. In short, we say that it is mass tight at time $x$ (or, $x$ is mass tight).
\end{defn}

\begin{defn}
    \label{def: Tight CR constraint}
    We say that the competitive ratio constraint is tight at $x$ if $\alpha_{f}(x) = OPT$. In short, we say that it is $CR$ tight at time $x$ (or, $x$ is $CR$ tight).
\end{defn}

\subsection{Zero Suffix and Zero Interval}
\label{sec:Zero Suffix and Zero Interval}

There are two special structures in a feasible solution $f$ that have nice properties, which we call the zero suffix and the zero interval. We explore such properties in Appendix~\ref{appendix:zero suffix and zero interval}.

\section{The Structure of Optimal Solutions}
\label{sec:structure of optimal solutions}

In this section, we will prove the structure of optimal solutions. We assume that the true OPT is given and answer the structural question we mentioned in Section~\ref{sec:introduction} for the optimal solutions. In Section~\ref{subsec: optimal solution by greedy algorithm} we will argue that there is an optimal solution by greedily placing mass without violating any mass or $CR$ constraints, as argued in Theorem~\ref{our contribution:thm 1, greedy gives optimal}. Moreover, we prove that every optimal solution must have the same prefix. In section~\ref{subsec: general structure of optimal solutions} we will explore structures of optimal solutions in large and small $\delta$ regimes, and prove the corresponding structure theorems, as stated in Theorem~\ref{our contribution: thm 2, general, large delta} to~\ref{our contribution:thm 5, closely related greedy, all delta}.

\subsection{Optimal Solution by the Greedy algorithm}
\label{subsec: optimal solution by greedy algorithm}

We prove that at time $1$ the mass constraint cannot be tight (since we show that its bad interval is empty) in Appendix~\ref{appendix:bad intervals in appendix} and that the $CR$ constraint at time $1$ must be tight (Theorem~\ref{thm: CR tight at time 1}) in any optimal solution.
A consequence of this result is that either the~$CR$ constraint or the mass constraint must be tight in $[0,1]$ (Theorem~\ref{thm: in any opt solution something is tight in [0,1]}). It turns out that the greedy algorithm, which places one unit of mass to make either~$CR$ or mass constraint tight at each point until it is not allowed to do so, gives an optimal solution to the problem (Theorem~\ref{our contribution:thm 1, greedy gives optimal}). 
Due to space constraints, full proofs of the lemmas and theorems stated in this section are deferred to Appendix~\ref{appendix:proofs of theorems in section 3.1 in appendix}.

\begin{restatable}{lemma}{FirstTightPointHasMass} \label{lem: first CR tight point has nonzero mass}
    Let $T \geq 0$ be some finite point that is not $CR$ tight.
    Then for any feasible solution $f$, if time~$t < \infty$ is the first $CR$ tight point after $T$, then there is nonzero mass at time $t$.
\end{restatable}

Based on Lemma~\ref{lem: first CR tight point has nonzero mass}, we can explore the tightness at time 1 in any optimal solution.

\begin{restatable}{thm}{CRtightAtOne}
    \label{thm: CR tight at time 1}
    Let $f$ be any optimal solution. Then the competitive ratio constraint is tight at time~1.
\end{restatable}

Now, we can use Theorem~\ref{thm: CR tight at time 1} to show that in any optimal solution, either the mass constraint or the $CR$ constraint is tight in $[0,1]$.

\begin{restatable}{thm}{SomethingTightInZeroToOne}
    \label{thm: in any opt solution something is tight in [0,1]}
    Let $f$ be any optimal solution. Then for any point in $[0,1]$, either the competitive ratio constraint or the mass constraint is tight.
\end{restatable}

Given Theorem~\ref{thm: in any opt solution something is tight in [0,1]}, we can use the same technique (by simple induction) as \cite{Dinitz2024ControllingTR} to prove the uniqueness in $[0,1]$ for optimal distributions. This is because for every optimal solution, $[0,1]$ is always either $CR$ or mass tight by Theorem~\ref{thm: in any opt solution something is tight in [0,1]} and we use the fact that the bad interval of $x \in [0,1]$ is a (sub)interval of its prefix.

\begin{restatable}{thm}{UniqueDistributionInZeroToOne}
    \label{thm: unique mass placement in [0,1]}
    If $f$ and $f'$ are optimal solutions, then $f_t = f'_t$ for all $0 \leq t \leq 1$.
\end{restatable}

After we show the structure in $[0,1]$ of an optimal solution, we examine its structure after time~1. The next lemma shows that $CR$ is non-increasing over zero intervals in any optimal solution.

\begin{restatable}{lemma}{CRcanNotIncreaseOverZeroIntervals}
    \label{lemma: CR cannot increase over zero intervals after time 1 in any opt solution}
    In any optimal solution, in any zero interval, the competitive ratio is non-increasing.
\end{restatable}

Now, let us introduce the definition of a tight prefix.

\begin{defn}
    \label{def: def of tight prefix}
    We say that $[0,t]$ is a tight prefix of a solution $f$ if every point in $[0,t]$ is either $CR$ tight or mass tight.
\end{defn}

Thus, $[0,1]$ is a tight prefix for any optimal solution $f$. 
One natural way to obtain a tight prefix is to use the \emph{greedy} algorithm---informally, we iterate through the times starting at time $0$ and greedily place mass to make either the $CR$ constraint or mass constraint tight at each time step. Technically, we do a special check at time right after point~1. For a formal definition of the greedy algorithm, see Definition~\ref{def: def of the greedy algorithm} in Appendix~\ref{appendix:structure proofs in appendix}.

Then, a natural question is to ask if the greedy algorithm gives an optimal distribution. First, one might worry that by placing probability mass greedily in $[0,1]$, there is actually some time $T > 1$ where the $CR$ already exceeds $OPT$.  Fortunately, we show that this cannot happen.

\begin{restatable}{lemma}{GreedyInZeroToOneIsSafe}
    \label{lem: greedy in [0,1] is safe for points passed 1}
    By placing mass greedily in $[0,1]$, there is no point $T>1$ where the competitive ratio already exceeds $OPT$.
\end{restatable}

Let $\tau>0$ be a sufficiently small time step to discretize the domain of the problem. As a corollary, we have the following.

\begin{restatable}{corollary}{ZeroMassAfterTimeOneWillNotIncreaseTheCR}
    \label{cor: zero mass at 1+tau will not increase the CR}
    After mass is assigned greedily in $[0,1]$, placing a zero mass at time $1+\tau$ will not increase the competitive ratio.
\end{restatable}

Corollary~\ref{cor: zero mass at 1+tau will not increase the CR} states that after greedily placing mass in $[0,1]$, if we place a zero mass at $1+\tau$, the $CR$ either stays flat or decreases. Since we know that the $CR$ constraint is tight at time $1$ (Theorem~\ref{thm: CR tight at time 1}), this means that we can either place nonzero mass at time $1+\tau$ (if the $CR$ would decrease if we place mass $0$) or we must place $0$ mass at time $1+\tau$ (if the $CR$ stays flat when placing $0$ mass).

\begin{defn}
    \label{def: def of two worlds}
    After placing mass greedily in $[0,1]$, we say that we are in \emph{world 1} if one has to place a zero mass at time $1+\tau$ and we are in \emph{world 2} if one can place some nonzero mass at time $1+\tau$.
\end{defn}

Note that world 1 and world 2 partition the solution space. In the following lemma, we show that in different world, $CR$ behaves differently in zero intervals after time 1.

\begin{restatable}{lemma}{BehaviorOfZeroIntervalsInTwoWorlds}
    \label{lem: CR behavior over zero intervals in two worlds}
    Let $f$ be any distribution and let $I$ be any zero interval after time 1. In world 1, the competitive ratio is flat in $I$. In world 2, the competitive ratio decreases in $I$.
\end{restatable}

Since solutions in world 1 and world 2 have different structures, the remaining part of Section~\ref{subsec: optimal solution by greedy algorithm} shows that the greedy algorithm gives an optimal solution separately. The easier case is when we are in world 1, since we are left with no further choices.  So we have the following lemma.

\begin{restatable}{lemma}{InWorldOneGreedyGivesOptimal}
    \label{lem: in world 1 greedy gives optimal}
    In world 1, the greedy algorithm gives an optimal solution.
\end{restatable}

When we are in world 2, it is not trivial to say that the greedy algorithm is optimal because we need to rule out the case that when the greedy algorithm has already placed a total of $\delta$ mass in the bad interval suffix $(L_b,\infty]$ (we show that in 
Appendix~\ref{appendix:bad intervals in appendix} there is at most $\delta$ mass in this suffix $(L_b,\infty]$ to get a feasible solution), the total mass is still less than 1. 
Therefore, to prove the final Theorem~\ref{our contribution:thm 1, greedy gives optimal}, we first prove the following lemma in world 2.

\begin{restatable}{lemma}{ExistenceOfAnOptSolutionWithLongestTightPrefix}
    \label{lem: existence of an opt solution with the longest tight prefix}
    In world 2, there exists an optimal solution of the following form:

    --- Either competitive ratio or mass constraint is tight for some prefix ending at time~$T$.

    --- Possibly nonzero mass is at time $T+\tau$.

    --- zero mass after time $T+\tau$. 
\end{restatable}

This lemma indicates that some optimal solutions are in an equivalence class, namely that they have the same length of tight prefix (which is the longest tight prefix among all optimal solutions), but after the tight prefix there might be a semi-flexible way to place mass by just not violating any constraint at any point outside the tight prefix. 

Now, we prove Theorem~\ref{our contribution:thm 1, greedy gives optimal}. Note that this is not trivial by Lemma~\ref{lem: existence of an opt solution with the longest tight prefix}. For example, after we make a tight prefix $[0,x]$ where $x$ is the first point whose bad interval is of the form $I_\gamma(x)=(\cdot,\infty]$, if there is still remaining mass to place, it could be the case that the longest tight prefix solution puts mass to $\infty$. Thus, every point in the zero suffix could be mass tight. However, greedy algorithm will place mass greedily, resulting in another optimal solution. We will show that the greedy algorithm constructs the solution which coincides with the solution in Lemma~\ref{lem: existence of an opt solution with the longest tight prefix} from 0 to $L_b$, and hence gives an optimal solution to the problem (Theorem~\ref{our contribution:thm 1, greedy gives optimal}).

\OurContributionTheoremONEGreedyGivesOptimal*

\begin{defn}
    \label{def: def of the greedy solution}
    We call the optimal solution constructed by the greedy algorithm in Theorem~\ref{our contribution:thm 1, greedy gives optimal} the \emph{greedy solution}.
\end{defn}

\subsection{General Structure of Optimal Solutions}
\label{subsec: general structure of optimal solutions}

In this section, we will use properties in Section~\ref{subsec: optimal solution by greedy algorithm} to examine the structures of optimal solutions. We show that there is a threshold such that the optimal solution is unique and has a zero suffix when $\delta$ is larger than that threshold (Theorem~\ref{our contribution: thm 2, general, large delta}), and there is not a unique optimal solution when $\delta$ is smaller than the threshold (Theorem~\ref{our contribution:thm 3, general, small delta}). In a small $\delta$ regime, we further show the structure of the greedy solution (Theorem~\ref{our contribution:thm 4, greedy solution, small delta}). Finally, we prove that there is an optimal solution with a finite support $[0,L_b] \cup \{\infty\}$ (Theorem~\ref{our contribution:thm 5, closely related greedy, all delta}). Full proofs of the lemmas and theorems stated in this section are deferred to Appendix~\ref{appendix:proofs of theorems in section 3.2 in appendix}.

To explore the general structure of optimal solutions, we first show that $CR$ must be tight in $[1,I_3]$ in any optimal solution.

\begin{restatable}{thm}{CRtightFromOneToIthree}
    \label{thm: CR must be tight in [1,I3]}
    For any optimal solution $f$, the competitive ratio constraint must be tight in $[1,I_3]$.
\end{restatable}

Next, we show that the structure in $(1,I_3]$ of any optimal solution is particularly simple, even when the optimal solution is in world 2.
Let $\tau>0$ be a sufficiently small time step (again we discretize the problem). The following lemma calculates the amount of mass at each point to keep $CR$ tight before $I_3$ in any optimal solution in world 2.

\begin{restatable}{lemma}{ExponentialToMakeCRtight}
    \label{lem: exponential function to make CR tight}
    To keep $CR$ tight in $(1,I_3]$, for any $k \geq 1, k \in \mathbb{N}$, we have $f_{1+k\tau} = (1+\tau)^{k-1}f_{1+\tau}$, where $f_{1+\tau}$ is calculated in Lemma~\ref{lem: mass at 1+tau to make CR tight} in Appendix~\ref{appendix:proofs of theorems in section 3.2 in appendix}.
\end{restatable}

Then, Theorem~\ref{thm: CR must be tight in [1,I3]} and Lemma~\ref{lem: exponential function to make CR tight} lead to the following Corollary.

\begin{corollary}
    \label{cor: structure of every opt solution before I3}
    Every optimal solution in the same world (world 1 or world 2) is the same in~$[0,I_3]$. Moreover, in world 2, the distribution in $[1,I_3]$ of every optimal solution is the same exponential function.
\end{corollary}

Now, we determine whether the optimal solution lies in world 1 or world~2 as $\delta$ varies. We first show that, in world 2, any optimal solution must have a $\delta$ mass in the bad interval suffix $(L_b,\infty]$.

\begin{restatable}{thm}{BadIntervalSuffixHasDeltaMass}
    \label{thm: bad interval suffix has exactly delta mass}
    For any optimal solution, if the competitive ratio is strictly decreasing over zero intervals after time 1 (i.e., world 2), then there is exactly $\delta$ mass in the bad interval suffix~$(L_b,\infty]$.
\end{restatable}

Then, we show that when $\delta < \frac{OPT-1}{\frac{1}{a}-1}$, every optimal solution must be in world 2. It is worth mentioning that this case could happen when $\delta$ is small enough because $OPT \in [\frac{e}{e-1+a},2-a]$ and then the implicit expression $\frac{OPT-1}{\frac{1}{a}-1}$ is bounded below by some value which is independent of $\delta$.

\begin{restatable}{thm}{SmallDeltaIsInWorldTwo}
    \label{thm:small delta is in world 2}
    If $\delta < \frac{OPT-1}{\frac{1}{a}-1}$, then every optimal solution must be in world 2.
\end{restatable}

Now we know that when $\delta$ lies in $[0,\frac{OPT-1}{\frac{1}{a}-1})$ (the small $\delta$ regime) we are in world 2, and when $\delta=1$ (no tail constraint) the classical solution is in world 1. We will prove that when $\frac{OPT-1}{\frac{1}{a}-1} \leq \delta \leq 1$ (the large $\delta$ regime), we are in world 1. That is, we prove Theorem~\ref{our contribution: thm 2, general, large delta}.

\OurContributionTWOLargeDelta*

Thus, in the large $\delta$ regime, the difference between the two-state with tail constraints setting and pure buying with tail constraints setting~\cite{Dinitz2024ControllingTR} is just that extra mass is added at $\infty$. However, when $\delta$ is decreasing and passing the threshold $\frac{OPT-1}{\frac{1}{a}-1}$, the structure of the optimal solution is fundamentally changed. 
By the fact that there are multiple optimal solutions in the small $\delta$ regime, the remaining part of this section will basically explore the structure of the \emph{greedy solution} in the small $\delta$ regime. 

By Theorem~\ref{thm: CR must be tight in [1,I3]} and Lemma~\ref{lem: exponential function to make CR tight}, we know that in the greedy solution, after time 1 the solution first has an exponential function to keep $CR$ tight. Then there are two possibilities: adding mass exponentially until $L_b$ does not violate any mass constraint (thus there is no mass tight point inside $[1,L_b]$), or at some point we must add less amount of mass to make the corresponding point mass tight. In the simulation Figure~\ref{fig: world 2 greedy with finite support} and Figure~\ref{fig: unique solution in world 1} one can see that both cases could happen. We show in Theorem~\ref{our contribution:thm 4, greedy solution, small delta} that in the greedy solution if the latter case happens, once some point has to become mass tight, it will keep mass tight until $L_b$. 

We formalize the structure theorem for the greedy solution in the small $\delta$ regime.

\OurContributionFOURGreedyStructure*

Note that although Theorem~\ref{our contribution:thm 4, greedy solution, small delta} states that the greedy solution always has finite support, the support can vary depending on different inputs $a,\delta$ and $\gamma$. However, based on the theorem, one can easily find another optimal solution with a finite support $[0,L_b] \cup \{\infty\}$ which is closely related to the greedy solution, by placing all the mass in $(L_b,\infty]$ to $\infty$. Together with the fact that in the large $\delta$ regime, there is a unique optimal solution with support $[0,1] \cup \{\infty\}$, we have a unified structure result for all $\delta$.

\OurContributionFiveAllDelta*

We will leverage the structure theorems to design two algorithms for computing optimal purchase distributions (Appendix~\ref{sec:algorithms}), and then present visualizations of the solution structures under various parameter settings (Appendix~\ref{sec:Simulations}).

\appendix

\section{Bad intervals}
\label{appendix:bad intervals in appendix}

Recall that the bad interval $I_\gamma(x)$ is defined to be
\[I_{\gamma}(x)=\{t:\alpha(t,x) > \gamma\}=\begin{cases} 
      \{t \leq x: \frac{t+1-a+a(x-t)}{x}> \gamma\},  & \text{for} \ 0 \leq x \leq 1 \\
      \{t \leq x: \frac{t+1-a+a(x-t)}{1-a+ax} > \gamma\} \cup \{t>x:\frac{x}{1-a+ax}>\gamma\},  & \text{for} \ x>1.
\end{cases}\]

If we try to simplify the above formula with a bit calculations, it is not hard to find that we have two different structures of the bad interval when $\gamma \geq \frac{1}{a}$ and when $2-a \leq \gamma < \frac{1}{a}$. When $\gamma \geq \frac{1}{a}$,
\[I_{\gamma}(x) = \begin{cases} 
      \{t : 0<t \leq x\},  & \text{for} \ 0 \leq x \leq \frac{1-a}{\gamma-a} \\
      \{t : \frac{\gamma-a}{1-a} x -1 <t \leq x\},  & \text{for} \ \frac{1-a}{\gamma-a} < x < \frac{1-a}{\gamma-1} \\
      \emptyset,  & \text{for} \ \frac{1-a}{\gamma-1} \leq x.
   \end{cases}\]
Note that in this case the bad interval is nonempty only in $[0,1]$. However, when $2-a \leq \gamma < \frac{1}{a}$, 
\[I_{\gamma}(x) = \begin{cases} 
      \{t : 0<t \leq x\},  & \text{for} \ 0 \leq x \leq \frac{1-a}{\gamma-a} \\
      \{t : \frac{\gamma-a}{1-a} x -1 <t \leq x\},  & \text{for} \ \frac{1-a}{\gamma-a} < x < \frac{1-a}{\gamma-1} \\
      \emptyset,  & \text{for} \ \frac{1-a}{\gamma-1} \leq x \leq \frac{(\gamma-1)(1-a)}{1-a\gamma} \\
      \{t : \frac{(\gamma-1)a}{1-a}x+\gamma-1<t \leq x\},  & \text{for} \ \frac{(\gamma-1)(1-a)}{1-a\gamma} < x \leq \frac{\gamma(1-a)}{1-a\gamma} \\
      \{t : \frac{(\gamma-1)a}{1-a}x+\gamma-1<t\},  & \text{for} \ \frac{\gamma(1-a)}{1-a\gamma} < x \\
      \{\infty\},  & \text{for} \ x=\infty.
   \end{cases}\]

In this case, it turns out that bad intervals are different when $x$ lies in five disjoint subintervals of $\mathbb{R}^{+}$. For simplicity, we call these subintervals of $x$ to be $BI_1$ to $BI_5$. When $x \in BI_1$, the corresponding bad interval is a prefix of $x$. When $x \in BI_2$, the corresponding bad interval is a subinterval of a prefix of $x$. When $x \in BI_3$, the bad interval of $x$ is empty. Note that $1 \in BI_3$, it implies that at time $1$ the mass constraint cannot be tight. When $x \in BI_4$, it is analogous to the structure of the bad interval of $x \in BI_2$. When $x \in BI_5$, the bad interval is part of its prefix plus a suffix stretches out to infinity. 
In particular, when $x = \frac{\gamma(1-a)}{1-a\gamma}$, we denote $L_b \coloneqq \frac{(\gamma-1)a}{1-a}x+\gamma-1 = \frac{\gamma-1}{1-a\gamma}$. Note that $L_b$ is the last point that is outside the bad interval of any point $x \in BI_5$, and we call $(L_b,\infty]$ the bad interval suffix. We also denote $I_3 \coloneqq \frac{(\gamma-1)(1-a)}{1-a\gamma}$ to be right boundary of $BI_3$. Note that $I_3 \geq 1$. We will use $I_3$ and $L_b$ in our analysis when we explore the structure of optimal solutions.

Our paper will focus on the more interesting case $2-a \leq \gamma < \frac{1}{a}$. In this setting, points after time 1 have nontrivial bad intervals, thus one can expect a nontrivial structure of optimal solutions in the two states setting.

\section{Definitions and results for zero suffix and zero interval}
\label{appendix:zero suffix and zero interval}

\subsection{Zero Suffix}

\begin{defn}
\label{def: def of a zero suffix}
    We say that a distribution $f$ has a zero suffix in $[1,\infty)$ if $\exists T \geq 1$ such that $f_t=0$ for every $t \in [T,\infty)$, where $[T,\infty)$ is called a zero suffix in $f$. 
\end{defn}

The following lemma shows the monotonicity of the competitive ratio function on a zero suffix in a distribution $f$.

\begin{lemma}
\label{lem: monotonicity of zero suffix}
    For any distribution $f$ with a zero suffix in $[1,\infty)$, the competitive ratio is a monotone function of $x$ on the zero suffix, and its limit when $x \rightarrow \infty$ equals the competitive ratio at $\infty$.
\end{lemma}

\begin{proof}
Let $[T,\infty)$ be a zero suffix of $f$. By Definition~\ref{def: def of a zero suffix}, we have $T \geq 1$. Thus for any~$x \in [T,\infty)$,
\begin{align*}
\alpha_f(x) =& \int_{t=0}^{x} { \frac{t+1-a+a(x-t)}{1-a+ax}f_t \, dt} + \frac{x}{1-a+ax} f_{\infty}\\
=& \int_{t=0}^{T} { \frac{t+1-a+a(x-t)}{1-a+ax}f_t \, dt} + \frac{x}{1-a+ax} f_{\infty}\\
=& \int_{t=0}^{T}{f_t \, dt} + \int_{t=0}^{T} { \frac{(1-a)t}{1-a+ax}f_t \, dt} + \frac{x}{1-a+ax} f_{\infty}\\
=& 1-f_{\infty} + \frac{1-a}{1-a+ax} \int_{t=0}^{T} { t f_t \, dt} + \frac{x}{1-a+ax} f_{\infty}.
\end{align*}
Note that for any given distribution $f$, $\int_{t=0}^{T} { t f_t \, dt}$ is a fixed value and we denote this value as constant $A$. Denote $B \coloneqq f_\infty$ which is also a constant. Then the above competitive ratio becomes
\[g(x) \coloneqq \alpha_f(x) = 1-B+\frac{1-a}{1-a+ax}A+\frac{x}{1-a+ax}B.\]
Taking the derivative of $g(x)$ implies that $g^{\prime}(x)=\frac{(1-a)(B-aA)}{(1-a+ax)^2}$, which is positive (negative) when the fixed threshold $B-aA$ is positive (negative), and is $0$ when $B-aA=0$. Thus,~$\alpha_f(x)$ is monotone on $[T,\infty)$.

To see that competitive ratio at $\infty$ equals the limit of the competitive ratio function when $x \rightarrow \infty$, note that $\frac{1-a}{1-a+ax} \rightarrow 0$ and $\int_{t=0}^{T} { t f_t \, dt}$ is a constant (thus bounded above), and $\frac{x}{1-a+ax} = \frac{1}{\frac{1-a}{x}+a} \rightarrow \frac{1}{a}$, hence $\alpha_f(x) \rightarrow 1-f_{\infty}+0+\frac{1}{a}f_\infty = 1+(\frac{1}{a}-1)f_{\infty} = \alpha_f(\infty)$.

This finishes Lemma~\ref{lem: monotonicity of zero suffix}.
\end{proof}

Two corollaries follows directly from Lemma~\ref{lem: monotonicity of zero suffix}.

\begin{corollary}
\label{cor: corollary one-to-all suffix}
For any distribution $f$ with a zero suffix $[T,\infty)$, if $CR$ is tight at some point $x \in (T,\infty)$, then $CR$ is tight at any point in $[T,\infty)$.
\end{corollary}

\begin{corollary}
\label{cor: zero suffix with non tight CR}
For any distribution $f$ with a zero suffix $[T,\infty)$, if $CR$ is not tight at $\infty$, then $CR$ is not tight at any $x \in (T,\infty)$.
\end{corollary}

\subsection{Zero Interval}

We saw that the competitive ratio has good properties over the zero suffix. A natural extension is to explore if similar properties hold over a general zero interval which is defined as follows.

\begin{defn}
\label{def: def of a zero interval}
    We say that a distribution $f$ has a zero interval in $[1,\infty)$ if $\exists T_1, T_2$ with $1 \leq T_1 < T_2 < \infty$, such that $f_t=0$ for every $t \in [T_1,T_2]$, where $[T_1,T_2]$ is called a zero interval of $f$.

\end{defn}

The following lemma shows that the same monotone property of the competitive ratio function on the zero interval of a distribution $f$ holds.

\begin{lemma}
\label{lem: monotonicity of zero interval }
    For any distribution $f$ with a zero interval in $[1,\infty)$, the competitive ratio is a monotone function of $x$ on the zero interval.
\end{lemma}

\begin{proof}
Let $I \coloneqq [T_1,T_2]$ be the zero interval of $f$. Also denote $P \coloneqq [0,T_1)$ and $S \coloneqq (T_2,\infty)$. Then by Definition~\ref{def: def of a zero interval}, we have $1 \leq T_1$. Thus for any $x \in [T_1,T_2]$,
\begin{align*}
\alpha_f(x) =& \int_P { \frac{t+1-a+a(x-t)}{1-a+ax}f_t \, dt} + \int_S {\frac{x}{1-a+ax}f_t \, dt}+ \frac{x}{1-a+ax} f_{\infty}\\
=& \int_P {f_t \, dt} + \frac{(1-a)}{1-a+ax} \int_P { t f_t \, dt} + \frac{x}{1-a+ax} (\int_S {f_t \,dt} + f_{\infty})
\end{align*}

Note that for any given distribution $f$, $\int_P {f_t \,dt}$, $\int_P { t f_t \, dt}$ and $\int_S {f_t \,dt}+f_\infty$ are fixed value and $\int_S {f_t \,dt}+f_\infty$ is the total mass after the zero interval. Hence, $\alpha_f(x)$ depends on the mass distribution in the prefix, and only the total mass in the suffix. For a given $f$, denote $c \coloneqq \int_S {f_t \,dt}+f_\infty$, and $E \coloneqq \int_P tf_t \,dt$. Hence $\int_P {f_t \,dt}=1-c$. Then the above competitive ratio function becomes
\[\alpha_f(x) = 1-c + \frac{1-a}{1-a+ax}E+\frac{x}{1-a+ax}c.\]
Taking the derivative of $\alpha_f(x)$ implies that $\alpha_f^{\prime}(x)=\frac{(1-a)(c-aE)}{(1-a+ax)^2}$, which is positive (negative) when the fixed threshold $c-aE$ is positive (negative), and is $0$ when $c-aE=0$. Thus,~$\alpha_f(x)$ is monotone on the zero interval $[T_1,T_2]$.

This finishes Lemma~\ref{lem: monotonicity of zero interval }.
\end{proof}

However, if we analogously define the zero interval in $[0,1]$, we have a strictly decreasing competitive ratio function of $x$ over zero intervals.

\begin{defn}
\label{def: def of zero interval from 0 to 1}
    We say that a distribution $f$ has a zero interval in $[0,1]$ if $\exists T_1,T_2 $ with $0 \leq T_1 < T_2 \leq 1$, such that $f_t=0$ for every $t \in [T_1,T_2]$, where $[T_1,T_2]$ is called a zero interval of $f$.

\end{defn}

\begin{lemma}
\label{lem: CR decreasing over zero interval}
    For any distribution $f$ with a zero interval in $[0,1]$, the competitive ratio is a strictly decreasing function of $x$ on the zero interval.
\end{lemma}

\begin{proof}
We use the same idea in the previous lemma. Let $I \coloneqq [T_1,T_2]$ be the zero interval of $f$. Also denote $P \coloneqq [0,T_1)$ and $S \coloneqq (T_2,\infty)$. Then by Definition~\ref{def: def of zero interval from 0 to 1}, we have $0 \leq T_1 < T_2 \leq 1$. Thus for any $x \in [T_1,T_2]$,
\begin{align*}
\alpha_f(x) =& \int_P { \frac{t+1-a+a(x-t)}{x}f_t \, dt} + \int_S {\frac{x}{x}f_t \, dt}+ \frac{x}{x} f_{\infty}\\
=& (a+\frac{1-a}{x}) \int_P {f_t \, dt} + \frac{1-a}{x} \int_P { t f_t \, dt} + (\int_S {f_t \,dt} + f_{\infty})
\end{align*}

Note that for any given distribution $f$, $\int_P {f_t \,dt}$, $\int_P { t f_t \, dt}$ and $\int_S {f_t \,dt}+f_\infty$ are fixed value and $\int_S {f_t \,dt}+f_\infty$ is the total mass after the zero interval. Hence, $\alpha_f(x)$ depends on the mass distribution in the prefix, and only the total mass in the suffix. For a given $f$, denote $c \coloneqq \int_S {f_t \,dt}+f_\infty$, and $E \coloneqq \int_P tf_t \,dt$. Hence $\int_P {f_t \,dt}=1-c$. Then the above competitive ratio function becomes
\[\alpha_f(x) = (a+\frac{1-a}{x})(1-c) + \frac{1-a}{x}E+c.\]
Taking the derivative of $\alpha_f(x)$ implies that $\alpha_f^{\prime}(x)=-\frac{(1-a)(1-c+E)}{x^2}<0$, which is always negative. Thus, $\alpha_f(x)$ is strictly decreasing on the zero interval $[T_1,T_2]$.

This finishes Lemma~\ref{lem: CR decreasing over zero interval}.
\end{proof}

Similar to Corollary~\ref{cor: corollary one-to-all suffix}, the following useful result follows directly from Lemma~\ref{lem: monotonicity of zero interval } and Lemma~\ref{lem: CR decreasing over zero interval}.

\begin{corollary}
\label{cor: corollary point tight to interval tight}
For any distribution $f$ with a zero interval $[T_1,T_2]$, if competitive ratio constraint is tight at some point $x \in (T_1,T_2)$, then any point in $[T_1,T_2]$ has tight competitive ratio constraint.
\end{corollary}

When we have a monotone decreasing $CR$ over a zero interval, one can show that the $CR$ is also monotone decreasing over an $\epsilon$-interval.\footnote{one can use the same technique of proof here for an $\epsilon$-interval. It turns out that the threshold is $c-aE-a\epsilon$. If $c-aE<0$, then $c-aE-a\epsilon<0$.} Here, an $\epsilon$-interval means that $f_t=\epsilon$ for any $t \in [T_1,T_2]$ and some $T_1,T_2$. The properties of zero suffix, zero interval and $\epsilon$-interval are crucial when we analyze the structure of the optimal solution.

To finish this section, we prove that the asymptotic property holds for any feasible distribution, not only for a zero suffix.

\begin{lemma}
    \label{lem: asymptotic property holds for any feasible distribution}
    For any feasible distribution, we have $\lim_{x \to \infty} \alpha_f(x) = \alpha_f(\infty)$.
\end{lemma}

\begin{proof}
Recall that at $x=\infty$, $\alpha_f(x) = 1+(\frac{1}{a}-1)f_\infty$. When $x \to \infty$,
\[\alpha_f(x) = \int_{t=0}^{x}{f_t \, dt} + \frac{(1-a)}{1-a+ax} \int_{t=0}^{x} { tf_t \, dt} + \frac{x}{1-a+ax} (\int_{t=x}^{\infty}{f_t \, dt} + f_\infty).\]
The first term goes to $1-f_\infty$ and the third term goes to $\frac{1}{a}f_\infty$ when $x \to \infty$. We claim that the second term goes to 0. In fact, since $\int_{t=0}^{\infty}{f_t \, dt} < \infty$, there exists some finite large $T_\epsilon$ such that $\int_{t=T_\epsilon}^{\infty}{ f_t \, dt}<\epsilon$, for any $\epsilon \to 0$. Thus $\int_{t=T_\epsilon}^{x}{ f_t \, dt}<\epsilon$, and then we have~$\int_{t=T_\epsilon}^{x}{ tf_t \, dt} \leq x \int_{t=T_\epsilon}^{x}{ f_t \, dt}<x \epsilon$. Then
\[ \frac{(1-a)}{1-a+ax} \int_{t=0}^{x} { tf_t \, dt} =  \frac{(1-a)}{1-a+ax} \int_{t=0}^{T_\epsilon} { tf_t \, dt} +  \frac{(1-a)}{1-a+ax} \int_{t=T_\epsilon}^{x} { tf_t \, dt}\]
where the former term goes to 0 as $x \to \infty$ and the latter term goes to 0 when $\epsilon \to 0$ as $\frac{(1-a)\epsilon}{a} \to 0$.
\end{proof}

\section{Proofs in Section~\ref{sec:structure of optimal solutions}}
\label{appendix:structure proofs in appendix}

In this appendix, we give proofs for all lemmas and theorems in Section~\ref{sec:structure of optimal solutions}.

\subsection{Proofs in Section~\ref{subsec: optimal solution by greedy algorithm}}
\label{appendix:proofs of theorems in section 3.1 in appendix}

\FirstTightPointHasMass*

\begin{proof}
    Suppose by a contradiction that there is no mass at time $t$. By the definition of~$t$, we know that the $CR$ at time $t-\tau$ is not tight, hence $CR$ is strictly increasing from $t-\tau$ to $t$, for sufficiently small time step~$\tau>0$. Without loss of generality, we can also assume at time $t-\tau, f_{t-\tau}=0$. This is because if there is mass at $t-\tau$ then it only increases the $CR$ at time $t-\tau$, hence the $CR$ goes down (from $t$ to $t-\tau$) means that it would still go down even if there is no mass at time $t-\tau$. Hence, $f_t=f_{t-\tau}=0$. It implies that $[t-\tau,t]$ is a subinterval of some zero interval and the $CR$ is strictly increasing over this subinterval. If there is no mass at time $t+\tau$, then $[t-\tau,t+\tau]$ is a zero interval, forcing that the $CR$ at $t+\tau$ becomes larger than OPT by the monotonic property of zero intervals, which is a contradiction. However, if there is mass at time $t+\tau$, its $CR$ will be even larger than when there is no mass at $t+\tau$, and hence still larger than OPT, which is a contradiction.

    Therefore, there is nonzero mass at time $t$.
\end{proof}

\CRtightAtOne*

\begin{proof}
    Suppose we have an optimal solution $f$ so that $CR$ is not tight at 1. Let us consider the first $CR$ tight point after 1, call it $t_0$, which may or may not exist. 
    
    First, suppose that such $t_0 < \infty$ exists. By Lemma~\ref{lem: first CR tight point has nonzero mass}, we have that $f_{t_0}$ is nonzero. Since time~1 is not in the bad interval of any other points, it is possible to move some tiny mass from $t_0$ to~1 and give us a feasible solution: we can move some extremely small amount of mass from~$t_0$ to time~1 to make the $CR$ not tight at time 1 (e.g., some mass to make the $CR$ at time~1 increase a half of the gap between OPT and the current $CR$). Our manipulation makes the $CR$ not tight at any $t_0 \leq t < \infty$ since moving mass from $t_0$ to some point earlier benefits the $CR$ at $t_0 \leq t < \infty$, however it may hurt the $CR$ at some point between 1 and $t_0$. Let $t_1$ be the first point with $1<t_1<t_0$ where the $CR$ becomes tight after our manipulation (if such $t_1$ exists). Again by Lemma~\ref{lem: first CR tight point has nonzero mass}, there is nonzero mass at~$t_1$. Similarly we move some mass from $t_1$ to~1 to make the $CR$ not tight at time~1, and repeat the procedure until any point $1 \leq t < \infty$ has non-tight $CR$ (we can do this in the discrete setting). It is important to note that after our mass moving argument (at least once by the definition of~$t_0$) it guarantees that there is no $CR$ convergent to OPT after time 1, namely that no point after 1 can have an arbitrarily close $CR$ to OPT. This observation makes us safe to finally move a sufficient small amount of mass from some point before time 1 to time 1, to keep feasibility. Now, since there is no $CR$ tight point passed~1, consider the first $CR$ tight point $t^0$ in the whole domain (i.e., the first $CR$ tight point after time 0). Lemma~\ref{lem: first CR tight point has nonzero mass} implies that $f_{t^0} \neq 0$ (if $t^0$ is the first point in the whole domain, we still need $f_{t^0} \neq 0$ to make it $CR$ tight). Hence, we can move some tiny mass from~$t^0$ to 1 in order to make the $CR$ not tight at any time $t \geq 1$ as there is no point whose $CR$ is arbitrarily close to OPT. At the same time, the $CR$ at any point $0 \leq t \leq 1$ decreases, since for any such point the best strategy is not to buy. Thus, we find a solution who has a strictly less OPT value, contradicts with the fact that $f_t$ is an optimal solution.

    Second, if $\infty$ becomes the first $CR$ tight point passed time 1, we can move some sufficiently small mass from $\infty$ to time 1, which leads to either some finite point passed~1 becomes $CR$ tight (where we have the first case) or all points passed 1 have non-tight $CR$ after moving all mass from $\infty$ (where we have the third case next).

    Third, if there is no $CR$ tight point (including at $\infty$) passed 1, all we need to show is that there is no $CR$ convergent (arbitrarily close) to OPT after time 1 as before, thus it is safe to move mass from the first tight point in the whole domain $t^0$ to time 1 and get a contradiction as we did previously. Since in this case no point has $CR$ tight passed 1 (including point 1), the convergence can only possibly happen at $\infty$. This contradicts with the asymptotic property we proved in Appendix~\ref{appendix:zero suffix and zero interval} that $CR$ at $x \to \infty$ converges to the $CR$ at $\infty$ which is not tight. Therefore, there is no arbitrarily close $CR$ to OPT after time 1, and we can go back to the first case which leads to a contradiction.

    This concludes that the competitive ratio must be tight at time 1 in any optimal solution.
\end{proof}

\SomethingTightInZeroToOne*

\begin{proof}
    Suppose by a contradiction that there exists $t_1 < 1$ so that neither is tight at $t_1$. Let $t_2>t_1$ be the first point such that something is tight at $t_2$. By Theorem~\ref{thm: CR tight at time 1} we know that the $CR$ is tight at $t=1$, so $t_2 \leq 1$. Then, let us move some sufficiently small amount of mass from $t_2$ to $t_1$. We will show that after this moving, we still get a feasible optimal solution, but in this optimal solution the competitive ratio is not tight at time 1, thus contradict with Theorem~\ref{thm: CR tight at time 1}.

    First of all, let us see $f_{t_2} \neq 0$. Note that $t_2$ is a point where something is tight. If $t_2$ is $CR$ tight, then since it is the first $CR$ tight point after $t_1$, by Lemma~\ref{lem: first CR tight point has nonzero mass}, there is nonzero mass at time $t_2$; If $t_2$ is mass tight and has zero mass, then since $t_2 \leq 1$, the bad interval structure implies that $t_2-\tau$ must also be mass tight, contradict with the fact that $t_2$ is the first something tight point after $t_1$. Thus, $f_{t_2} \neq 0$.

    Then, let us see the feasibility of our distribution after moving a sufficiently small amount of mass from $t_2$ to $t_1$. For those points $x$ where $t_1 \not\in I_{\gamma}(x)$, the total mass over $I_{\gamma}(x)$ cannot increase, in this case our new solution is feasible. For those points $x$ where $t_1 \in I_{\gamma}(x)$, if $t_2$ is also in $I_{\gamma}(x)$, then moving mass from $t_2$ to $t_1$ does not increase the overall mass over the bad interval, and therefore in this case our new solution is also feasible. Finally for those points $x$ where $t_1 \in I_{\gamma}(x)$ but $t_2 \not\in I_{\gamma}(x)$, since $t_1 < t_2 \leq 1$, by the structure of the bad interval, we know that $t_1 \leq x < t_2$. But note that $t_2$ is the first something tight point after~$t_1$, meaning that mass is not tight at $x$. Therefore moving a sufficiently small amount of mass to somewhere in $I_{\gamma}(x)$ can still make it feasible, as $[t_1,t_2]$ is compact.

    Finally, let us see the optimality of our new solution. For any $x < t_1$, $\alpha_f(x)$ does not change since $\alpha_f(x)$ depends only on the total mass after $x$ and does not depend on the mass distribution after $x$. For any $t_1 \leq x < t_2$, since neither constraint is tight at $x$, moving a sufficiently small amount of mass from $t_2$ to $t_1$ can still keep the $CR$ at $x$ not exceed OPT. For any $t_2 \leq x < \infty$, our moving of mass benefits the $CR$ at $x$, hence we have strictly decreasing~$CR$ at $x$.

    Thus after moving mass from $t_2$ to $t_1$, we still get a feasible optimal solution. In particular, since $t_2 \leq 1$, the $CR$ goes down at time 1, which contradicts with Theorem~\ref{thm: CR tight at time 1} that $CR$ must be tight at 1 for any optimal solution.
\end{proof}

\UniqueDistributionInZeroToOne*

\begin{proof}
    Let $\tau > 0$ be a sufficiently small time step to discretize the domain of the problem. Let $f^*$ be an optimal solution and $f: (0,1] \to \mathbb{R}$ be a function whose values are identical to the function values in $(0,1]$ defined in Definition~\ref{def: def of the greedy algorithm}. We will prove by induction that $f^*_i=f_i$ for all $i=\tau, 2\tau, \dots, 1$, which clearly implies the theorem.

    For the base case, it is not hard to see that $f_\tau = \min (\delta, \frac{OPT-1}{1-a}\tau)$ is the only possible feasible value which satisfies Theorem~\ref{thm: in any opt solution something is tight in [0,1]} for time $\tau$. To see this, first suppose that $f^*_\tau > \delta$. Then since $\tau \in I_{\gamma}(\tau)$ we have that $\sum_{t \in I_\gamma(\tau)}f^*_t > \delta$, which implies that $f^*$ is not feasible. This is a contradiction, and hence $f^*_\tau \leq \delta$. Similarly, suppose that $f^*_\tau > \frac{OPT-1}{1-a}\tau$. Then $\alpha_{f^*}(\tau) = \frac{\tau +1-a}{\tau} f^*_\tau + 1 -f^*_\tau = \frac{1-a}{\tau}f^*_\tau + 1 > OPT$, which contradicts the optimality of $f^*$. Hence, $f^*_\tau \leq f_\tau$. On the other hand, suppose that $f^*_\tau < f_\tau$. Then $f^*_\tau < \delta$ and $\alpha_{f^*}(\tau) < OPT$, which contradicts Theorem~\ref{thm: in any opt solution something is tight in [0,1]} for $x=\tau$. So, $f^*_\tau \geq f_\tau$, and thus $f^*_\tau = f_\tau$.

    For the inductive case, suppose that $f^*_t = f_t$ for all $t<j$. If $f^*_j> \delta - \sum_{t \in I_\gamma(j) \setminus \{j\}}f_t$, then 
    \[\sum_{t \in I_\gamma(j)}f^*_t = f^*_j + \sum_{t \in I_\gamma(j) \setminus \{j\}}f_t > \delta - \sum_{t \in I_\gamma(j) \setminus \{j\}}f_t + \sum_{t \in I_\gamma(j) \setminus \{j\}}f_t = \delta\]
    which contradicts the feasibility of $f^*$. So, $f^*_j \leq \delta - \sum_{t \in I_\gamma(j) \setminus \{j\}}f_t$.

    Similarly, suppose that $f^*_j > \frac{j}{1-a}(OPT-1) - \sum_{t<j}(1+t-j)f_t$. Then we have that 
    \begin{align*}
    \alpha_{f^*}(j) =& \sum_{t \leq j} \frac{t+1-a+aj-at}{j}f^*_t + 1 - \sum_{t \leq j} f^*_t \\
    =& 1+\left( \frac{j+1-a}{j} -1 \right) f^*_j + \sum_{t \leq j-\tau} \left( \frac{t+1-a+aj-at}{j} -1 \right) f^*_t \\
    =& 1+ \frac{1-a}{j} f^*_j + \sum_{t \leq j-\tau}  \frac{t+1-a+aj-at-j}{j} f_t  \\
    >& 1+(OPT-1) - \frac{1-a}{j} \sum_{t \leq j-\tau} (1+t-j)f_t + \sum_{t \leq j-\tau}  \frac{t+1-a+aj-at-j}{j} f_t  \\
    =& OPT,
    \end{align*}
    where the third equality is by the induction hypothesis. This is a contradiction since $f^*$ is optimal. Hence, $f^*_j \leq \frac{j}{1-a}(OPT-1) - \sum_{t<j}(1+t-j)f_t$.

    Now, suppose that $f^*_j<f_j$. Then following the above series of inequalities but switching the $>$ to $<$ implies that $\alpha_{f^*}(j)< OPT$. Also, 
    \[\sum_{t \in I_\gamma(j)}f^*_t = f^*_j + \sum_{t \in I_\gamma(j) \setminus \{j\}}f_t < \delta - \sum_{t \in I_\gamma(j) \setminus \{j\}}f_t + \sum_{t \in I_\gamma(j) \setminus \{j\}}f_t = \delta.\]
    But now we have a contradiction to Theorem~\ref{thm: in any opt solution something is tight in [0,1]} for $x=j$. Thus, $f^*_j \geq f_j$. Hence $f^*_j = f_j$ in $[0,1]$, as claimed.
\end{proof}

\CRcanNotIncreaseOverZeroIntervals*

\begin{proof}
We know by Lemma~\ref{lem: CR decreasing over zero interval}, in any optimal solution, in any zero interval before time 1, the $CR$ is non-increasing. Thus, it suffices to show the statement holds for any zero interval after time 1.

    We know from Theorem~\ref{thm: CR tight at time 1} that the $CR$ constraint must be tight at time 1 in any optimal solution $f$. Let $\tau >0$ be a sufficiently small time step. Suppose that there is a zero mass at time $1+\tau$. Then, the competitive ratio at time $1+\tau$ cannot increase (compared with the $CR$ at time 1) in any optimal solution as there are no other choices (since at time 1, $CR$ must equal to OPT in any optimal solution). Hence, over this zero interval (one point of zero mass), we cannot have increasing $CR$. In other words, $c-aE \leq 0$ in Lemma~\ref{lem: monotonicity of zero interval }. That is, if a zero interval is right after time 1, the total mass after the zero interval minus $a$ times the expectation before the zero interval is non-positive. Note that points in any other zero interval after time 1 can only have smaller (or equal) $c$ and larger (or equal)~$aE$ compared to those at time $1+\tau$. Therefore, the $CR$ cannot increase over any zero interval after time 1.
\end{proof}

\begin{defn}
    \label{def: def of the greedy algorithm}
    An algorithm is called greedy if it places mass to make either the competitive ratio constraint or the mass constraint tight at each point until it is not allowed to do so. 
    Formally, for sufficiently small time step $\tau>0$, let $L_3$ and $I_3$ be the left and right endpoints of the bad interval $BI_3$, $L_5$ be the left endpoint of $BI_5$, the greedy algorithm is the following.
    \begin{itemize}
        \item $f_\tau = min(\delta, \frac{OPT-1}{1-a} \tau).$
        \item For $x=2\tau, 3\tau, \dots, L_3$, $f_x = \min \{ \delta - \sum_{t \in I_{\gamma}(x)\setminus \{x\}} f_t ,  \frac{x}{1-a} (OPT-1)-\sum_{t<x} (1+t-x) f_t\}$.
        \item For $x=L_3+\tau,\dots,1$, $f_x = \frac{x}{1-a} (OPT-1)-\sum_{t<x} (1+t-x) f_t$.
        \item Check if $OPT = \frac{1+\tau}{1-a+a(1+\tau)}+\sum_{t \leq 1}\frac{(1-a)(t-\tau)}{1-a+a(1+\tau)}f_t$. That is, check whether $CR$ stays constant when adding zero mass at $1+\tau$. If yes, set $f_\infty = \frac{OPT-1}{\frac{1}{a}-1}$ and return $f$. Otherwise, continue to the next step.
        \item For $x=1+\tau,\dots,I_3$, let $m_x = \frac{1-a+ax}{1-a}OPT-\frac{x}{1-a} - \sum_{t<x}(1+t-x) f_t.$ \begin{itemize}
            \item If $1- \sum_{t < x}f_t \leq m_x$, set $f_x = 1- \sum_{t < x}f_t$. Return $f$.
            \item Otherwise, $f_x=m_x$.
        \end{itemize}
        \item For $x=I_3+\tau,\dots, L_5$, let $m_x = \min \{\delta - \sum_{t \in I_{\gamma}(x)\setminus \{x\}} f_t, \frac{1-a+ax}{1-a}OPT-\frac{x}{1-a} - \sum_{t<x}(1+t-x) f_t\}$. \begin{itemize}
            \item If $1- \sum_{t < x}f_t \leq m_x$, set $f_x = 1- \sum_{t < x}f_t$. Return $f$.
            \item Otherwise, $f_x=m_x$.
        \end{itemize}
        \item For $x \geq L_5+\tau$, let $m_x = \min \{\delta - \sum_{t=L_b+\tau}^{x-\tau} f_t, \frac{1-a+ax}{1-a}OPT-\frac{x}{1-a} - \sum_{t<x}(1+t-x) f_t\}$. \begin{itemize}
            \item If $1- \sum_{t < x}f_t \leq m_x$, set $f_x = 1- \sum_{t < x}f_t$. Return $f$.
            \item Otherwise, $f_x=m_x$.
        \end{itemize}
    \end{itemize}
\end{defn}

We note that when $x \geq L_5+\tau$, the formula to make mass tight at point $x$ changes. As we discussed, this is because of the different structure of the bad intervals: the algorithm has to make sure that the first point in $BI_5$ must always satisfy the mass constraint when adding mass at every point $x$. Moreover, we will argue that the greedy algorithm will terminate and hence has a finite support, in the proof of Theorem~\ref{our contribution:thm 1, greedy gives optimal}.

\GreedyInZeroToOneIsSafe*

\begin{proof}
    By Theorem~\ref{thm: unique mass placement in [0,1]}, we know that every optimal solution is the same in $[0,1]$. By Theorem~\ref{thm: in any opt solution something is tight in [0,1]} and Definition~\ref{def: def of the greedy algorithm}, we know that the greedy algorithm has the same mass distribution in $[0,1]$ as the optimal solution. By the formula of $CR$, we know that the $CR$ at $x$ only depends on the mass distribution before $x$. Therefore, by placing probability mass greedily in $[0,1]$, there is no point $T > 1$ where the $CR$ already exceeds $OPT$ --- otherwise the optimal solution is infeasible.
\end{proof}

\ZeroMassAfterTimeOneWillNotIncreaseTheCR*

\begin{proof}
    This corollary holds by taking $T=1+\tau$ in Lemma~\ref{lem: greedy in [0,1] is safe for points passed 1} and by the fact that the $CR$ at $x$ only depends on the mass distribution before $x$.
\end{proof}

\BehaviorOfZeroIntervalsInTwoWorlds*

\begin{proof}
    In world 1, by Definition~\ref{def: def of two worlds}, one has to place only zero mass at time $1+\tau$. Given the bad interval structure at time $1+\tau$, placing sufficiently small amount of mass does not violate the mass constraint. In other words, in world 1, any nonzero mass at time $1+\tau$ will cause the $CR$ at $1+\tau$ exceeding $OPT$. Since adding mass at $x$ only will increase the $CR$ at $x$, it implies that the $CR$ at $1+\tau$ equals $OPT$ by placing a zero mass at $1+\tau$ in world 1. Thus, we must have zero mass at finite points after $1+\tau$. By Lemma~\ref{lem: monotonicity of zero suffix}, we know that $CR$ is flat over $I$.

    In world 2, we are allowed to place nonzero mass at time $1+\tau$, meaning that if we put a zero at time $1+\tau$, the $CR$ will decrease. Therefore, $c-aE<0$ at time $1+\tau$ in Lemma~\ref{lem: monotonicity of zero interval }. By the fact that any point after time $1+\tau$ can only have smaller (or equal) $c$ and larger (or equal) $aE$ compared to those at time $1+\tau$, it turns out that $CR$ only decreases over $I$ by Lemma~\ref{lem: monotonicity of zero interval }.
\end{proof}

\InWorldOneGreedyGivesOptimal*

\begin{proof}
    In world 1, a zero mass at $1+\tau$ makes its $CR$ equals $OPT$. By Lemma~\ref{lem: CR behavior over zero intervals in two worlds}, $CR$ is flat over zero intervals after time 1. Due to Theorem~\ref{thm: unique mass placement in [0,1]}, any optimal solution has the same mass distribution in $[0,1]$ as the greedy algorithm placed, it turns out that there is a zero suffix in any optimal solution in world 1. Combined with Lemma~\ref{lem: monotonicity of zero suffix}, it implies that any optimal solution must be $CR$ tight at time $\infty$. In other words, $f_\infty = \frac{OPT-1}{\frac{1}{a}-1}$. This is what exactly the greedy algorithm does by Definition~\ref{def: def of the greedy algorithm}. Hence, in world 1, greedy algorithm gives an optimal solution. Furthermore, by our analysis, this solution is unique.
\end{proof}

\ExistenceOfAnOptSolutionWithLongestTightPrefix*

\begin{proof}
    Note that since there is a long suffix in our domain, there is no reason to believe that we have a unique optimal solution. Instead, let us consider the set of optimal solutions so that any solution in the set has the longest tight prefix among all optimal solutions. Take one optimal solution from this set, call it $f$. Also denote the longest tight prefix as $[0,T]$. If there is possibly nonzero mass at time $T+\tau$ and zero mass after time $T+\tau$, then $f$ is our required solution and we are done. If not, there exists some nonzero mass after time $T+\tau$. Let us then consider the first point $t_1>T+\tau$ with a nonzero mass and move some mass from $t_1$ to $T+\tau$. We claim that this will not violate any constraint. For the mass constraint, for any~$x$ such that $T+\tau \in I_\gamma(x)$, we can have either $I_\gamma(x) = (\cdot, x]$, or $I_\gamma(x)=[\cdot,\infty)$. For the latter case, if $T+\tau \in I_\gamma(x)$, then we must have $t_1 \in I_\gamma(x)$. For the former case, if $t_1<x$, then $T+\tau, t_1 \in I_\gamma(x)$. If $T+\tau \leq x < t_1$, we know that $t_1$ is the first nonzero mass point. This implies that $\sum_{t \in I_\gamma(x)}f_t \leq \sum_{t \in I_\gamma(T+\tau)}f_t < \delta$. Hence, one is able to move some amount of mass from $t_1$ to the point $T+\tau$ in $f$ by the compactness of $[T+\tau,t_1]$. This implies that our mass moving can only make $T+\tau$ become the first mass tight point. Now given feasibility, let us see how the $CR$ changes. Recall that $CR$ at point $x$ can be rewritten as only what happens beforehand. For any $x<T+\tau$, our mass moving will not change the $CR$ at time~$x$. Also recall that we are in world~2, meaning that over the zero interval the $CR$ is strictly decreasing. Hence any point~$x$ in the zero interval $[T+2\tau,t_1-\tau]$ cannot have tight $CR$. Moreover, our mass moving from~$t_1$ to $T+\tau$ can only benefit the $CR$ at $x \geq t_1$. Therefore, the only point whose $CR$ increases is the point $T+\tau$. Above all, either point $T+\tau$ becomes tight, or we run out of mass at time $t_1$. In the former case, we construct an optimal solution with a longer tight prefix than~$f$, which is a contradiction to the definition of $f$. In the latter case, we consider the next nonzero mass point and repeat the procedure, and finally get an optimal solution of the required form. 
\end{proof}

\OurContributionTheoremONEGreedyGivesOptimal*

\begin{proof}

    We denote the solution from Lemma~\ref{lem: existence of an opt solution with the longest tight prefix} as $f^p$, and the greedy algorithm returns a distribution $f^g$. We will prove that $f^g_x=f^p_x$ for each $x$ up to $L_b$ by induction.

    We have to argue that making something tight greedily without knowing the future keeps the points tight up to $L_b$. We know that $CR$ can be rewritten as only what happens beforehand. At the beginning, the bad interval of $x$ is in the form $I_\gamma(x) = (\cdot, x]$ and later it becomes $I_\gamma(x)=(\cdot,\infty]$. Any point before $L_b$ does not lie in the bad interval of $x$ in the second form, and for the first form there is no choice about which mass one can place at the next time to make something tight (the way of assigning mass is unique). If $T<L_b$, given the uniqueness of greedily placing mass, until one cannot make any point tight  (meaning that we place all the mass at the next place and nothing is tight), we know that is the special point $T+\tau$ and two solutions both have the same leftover mass. This says that, if points are tight from 0 up until some point $t$, then to make the next point be tight, we have no choice what to place. Let us formalize this.

    Note that $f^g$ does not violate any constraint. For the base case, consider $x=\tau$ and we must have $f_\tau^g \leq f_\tau^p = min(\delta, \frac{OPT-1}{1-a} \tau)$. If $f_\tau^g < f_\tau^p$, then at time $\tau$ neither is tight in~$f^g$. This violates the definition of greedy algorithm.

    For $\tau<x\leq L_b$, assume for induction that $f_t^g=f_t^p$ for all $t<x$. First, $f_x^g \leq f_x^p$. This is because on the one hand $f^g$ satisfies the mass constraint hence $f^g_x$ is less than the mass at $x$ in $f^p$ to make mass tight at $x$, thanks to the form of bad interval $I_\gamma(x) = (\cdot, x]$ when $x \leq L_b$. On the other hand, the $CR$ constraint at $x$ consists of terms for $t<x$ and the term at $x$. By the induction hypothesis, $f_x^g=f_x^p$ for the summation over $t<x$. Thus we find the required upper bound for the $f_x^g$ term by the feasibility. Again, if $f_x^g<f_x^p$, it violates the definition of greedy. Thus $f_x^g=f_x^p$ for each $x \leq L_b$.

    To finish the proof, since $f^p$ is optimal, the total mass equals 1. Also feasibility of $f^p$ implies that the total mass over $(L_b,\infty]$ in $f^p$ is no more than $\delta$, and thus there are at least $1-\delta$ mass over $[0,L_b]$ in $f^p$. Because $f^p=f^g$ in $[0,L_b]$, we know that there are at least $1-\delta$ mass in $[0,L_b]$ in $f^g$. Thus, it is impossible that the total mass is less than 1 if the algorithm places a total of $\delta$ mass in $(L_b,\infty]$. Note that it is allowed to place a $\delta$ mass in $(L_b,\infty]$ in a greedy way when we are in world 2, because a decreasing mass in this suffix to make the total mass in $(L_b,\infty]$ converge before $\delta$ cannot make $CR$ tight at the points (to make points $CR$ tight we should place more and more mass by the monotonicity property of the $\epsilon$-interval, by the fact that $CR$ only decreases in a zero interval, discussed in Appendix~\ref{appendix:zero suffix and zero interval}). Thus, the greedy algorithm gives an optimal solution to the problem.
\end{proof}

\subsection{Proofs in Section~\ref{subsec: general structure of optimal solutions}}
\label{appendix:proofs of theorems in section 3.2 in appendix}

\CRtightFromOneToIthree*

\begin{proof}
    Since the bad interval of any point $t \in [1,I_3]$ is empty, this theorem is equivalently to say that something must be tight in $[1,I_3]$. We prove this by contradiction.

    Suppose that neither is tight at some point $t_1 \in [1,I_3]$. Consider the first point $t_2$ after time~$t_1$ so that $CR$ is tight at $t_2$. If such $t_2$ does not exist, we have that any point is not~$CR$ tight after time $t_1$. If such $t_2$ exists, then we move a sufficiently small amount of mass (exists since~$t_2$ is $CR$ tight) from $t_2$ to $t_1$. Clearly, mass constraint holds as $t_1$ is not in the bad interval of any other point. For the $CR$ constraint, by the definition of $t_2$, $CR$ is not tight at any point in $[t_1,t_2)$, so moving a small amount of mass from $t_2$ to $t_1$ can keep $CR$ non-tightness at any point in $[t_1,t_2)$, as well as decreasing the $CR$ at any point in $[t_2,\infty)$. Now, any point in $[t_1,\infty)$ is not $CR$ tight. For point $t \in (1,t_1)$ we can move, one by one from right to left, some little mass to time $t_1$ so that no point is $CR$ tight after $t$ in each phase. Finally, we move some little mass from the first point to $t_1$, we get a solution with a less value of OPT, which gives us a contradiction.
\end{proof}

To prove Lemma~\ref{lem: exponential function to make CR tight}, we need the following observation. Let $\tau>0$ be a sufficiently small time step. Given that $CR$ is tight at time 1 (Theorem~\ref{thm: CR tight at time 1}), the next observation calculates the amount of mass at time $1+\tau$ to make it $CR$ tight.

\begin{lemma}
    \label{lem: mass at 1+tau to make CR tight}
    For $\tau>0$ small, to make $CR$ tight at time $1+\tau$, the amount of mass at time~$1+\tau$ is exactly $f_{1+\tau}=\tau (\sum_{t=0}^{1}f_t -1+a\sum_{t=0}^{1}tf_t)$.
\end{lemma}

\begin{proof}
    Since we are considering discrete time points, $CR$ is tight at time 1 implies that $\sum_{t=0}^{1}tf_t = \frac{OPT-1}{1-a}$. Then to make $CR$ tight at time $x= 1+\tau$ we have $OPT$ equals
     \begin{align*}
    \alpha_f(x) =& \sum_{t=0}^{1+\tau} \frac{t+1-a+ax-at}{1-a+ax}f_t + \sum_{t=1+2\tau}^{\infty} \frac{x}{1-a+ax}f_t \\
    =& \sum_{t=0}^{1+\tau} f_t + \frac{1-a}{1-a+ax} \sum_{t=0}^{1+\tau}tf_t + \frac{x}{1-a+ax}\sum_{t=1+2\tau}^{\infty}f_t \\
    =& \sum_{t=0}^{1} f_t + f_{1+\tau} + \frac{1-a}{1-a+a(1+\tau)}(\sum_{t=0}^{1}tf_t+(1+\tau)f_{1+\tau})+\frac{1+\tau}{1-a+a(1+\tau)}\sum_{t=1+2\tau}^{\infty}f_t \\
    =& \sum_{t=0}^{1} f_t + f_{1+\tau} + \frac{OPT-1}{1-a+a(1+\tau)} + \frac{(1-a)(1+\tau)}{1-a+a(1+\tau)} f_{1+\tau}+\frac{1+\tau}{1-a+a(1+\tau)}(1-\sum_{t=0}^{1}f_t-f_{1+\tau})
    \end{align*}
    Simplify the above equation, we derived that $f_{1+\tau} = \frac{\tau}{1-a}[(1-a)\sum_{t=0}^{1}f_t -1+aOPT]$. Note that we are in world 2 and this term turns out to be positive, as $c<aE = \frac{OPT-1}{\frac{1}{a}-1}$ and hence $\sum_{t=0}^{1}f_t=1-c>1-\frac{OPT-1}{\frac{1}{a}-1}=\frac{1-aOPT}{1-a}$. We use $OPT= 1+ (1-a)\sum_{t=0}^{1}tf_t$ again to get another form of $f_{1+\tau}$, which is $f_{1+\tau}=\tau (\sum_{t=0}^{1}f_t -1+a\sum_{t=0}^{1}tf_t)$ as required.
\end{proof}

Because of the fact that 
\[f_{1+\tau}=\tau (\sum_{t=0}^{1}f_t -1+a\sum_{t=0}^{1}tf_t) \leq \tau (\sum_{t=0}^{1}f_t -1+a\sum_{t=0}^{1}f_t) = \tau ((1+a)\sum_{t=0}^{1}f_t -1) \leq \tau(1+a-1)=\tau a.\]
Let $\tau$ be small so that $\tau a <\delta$, then $f_{1+\tau}<\delta$. It turns out that we start with $CR$ tight at time $1+\tau$. 

Now, let us prove Lemma~\ref{lem: exponential function to make CR tight}.

\ExponentialToMakeCRtight*

\begin{proof}
    Let us prove by induction. For $k=1$ the statement is trivial. Assume that it is true for $1,2,\dots,k-1$ and let us prove it holds for $k$. Because of the fact that $CR$ is tight at both~$k$ and $k-1$, we have
    \begin{align*}
    0 =& \alpha_f(1+k\tau)-\alpha_f(1+(k-1)\tau) \\
    =& [\sum_{t=0}^{1+k\tau}f_t + \frac{1-a}{1+ka\tau} \sum_{t=0}^{1+k\tau} tf_t + \frac{1+k\tau}{1+ka\tau} - \frac{1+k\tau}{1+ka\tau} \sum_{t=0}^{1+k\tau}f_t] \\
    &-  [\sum_{t=0}^{1+(k-1)\tau}f_t + \frac{1-a}{1+(k-1)a\tau} \sum_{t=0}^{1+(k-1)\tau} tf_t + \frac{1+(k-1)\tau}{1+(k-1)a\tau} - \frac{1+(k-1)\tau}{1+(k-1)a\tau} \sum_{t=0}^{1+(k-1)\tau}f_t] \\
    =& \frac{1-a}{1+ka\tau}f_{1+k\tau} - \frac{(1-a)\tau}{(1+ka\tau)(1+(k-1)a\tau)}(\sum_{t=0}^{1+(k-1)\tau}f_t + a\sum_{t=0}^{1+(k-1)\tau}tf_t)+\frac{(1-a)\tau}{(1+ka\tau)(1+(k-1)a\tau)}.
    \end{align*}
    $\sum_{t=0}^{1+(k-1)\tau}f_t = \sum_{t=0}^{1}f_t + \frac{(1+\tau)^{k-1}-1}{\tau}f_{1+\tau}$ and $\sum_{t=0}^{1+(k-1)\tau}tf_t = \sum_{t=0}^{1}tf_t+(k-1)(1+\tau)^{k-1}f_{1+\tau}$ by the induction hypothesis, thus
    \[f_{1+k\tau}=\frac{\tau}{1+(k-1)a\tau}(\sum_{t=0}^{1+(k-1)\tau}f_t + a\sum_{t=0}^{1+(k-1)\tau}tf_t-1) = (1+\tau)^{k-1}f_{1+\tau}.\]
\end{proof}

\BadIntervalSuffixHasDeltaMass*

\begin{proof}
    Suppose not, then $BI_5$ can receive mass as the mass constraint for any point in $BI_5$ is not tight. Note that the last finite $CR$ tight point exists. This is because if not, the bad interval suffix has $CR$ tight infinitely often, but sum up to some fixed mass bounded by $\delta$. Thus $f_t \to 0$ as $t \to \infty$. But to keep the $CR$ tight infinitely often, we need more mass for later points. This is because we are in world 2, and the $CR$ decreases over an $\epsilon$-interval (an interval where each point has an $\epsilon$ mass) as over a zero interval, thus gives a contradiction. We have the last $CR$ tight point $t$.

    Now we can move sufficiently small amount of mass from $t$ to some far enough point in the suffix, and repeat for the new final tight point in $(1,\infty)$ until there is no $CR$ tight point in $(1,\infty)$ and the bad interval suffix still has less than $\delta$ mass. Finally, we move sufficiently small amount of mass from the very first tight point in our domain to some far enough point in the suffix to make no point being $CR$ tight, as well as keeping feasibility in the bad interval suffix. This gives us a solution with a $CR$ value strictly less than OPT, which is a contradiction.
\end{proof}

\SmallDeltaIsInWorldTwo*

\begin{proof}
Recall that $CR$ in a zero interval after time 1 can be either flat or decreasing in any optimal solution, depending on the value of $c-aE$.
As we discussed, when $c-aE=0$ (in Lemma~\ref{lem: monotonicity of zero interval }) the $CR$ function is flat while when $c-aE <0$ the $CR$ function is dropping, it turns out that we are in world 1 if and only if $CR$ is tight at $\infty$, by the asymptotic result for any feasible distribution $f$ (in Appendix~\ref{appendix:zero suffix and zero interval}). Recall that $CR$ is tight at $\infty$ if $f_\infty = \frac{OPT-1}{\frac{1}{a}-1}$, it implies that if $\delta < \frac{OPT-1}{\frac{1}{a}-1}$, $CR$ is not tight at $\infty$ and hence we must be in world 2. 
\end{proof}

\OurContributionTWOLargeDelta*

\begin{proof}
    By previous results, the only thing we left to show is that when $\frac{OPT-1}{\frac{1}{a}-1} \leq \delta \leq 1$, the competitive ratio function keeps flat after time~1 with the value of $OPT$. Suppose this is not true. Then we are in world 2 for some optimal solution $f$ with input values $\{a,\{\delta_i\},\{\gamma_i\}\}$. By Theorem~\ref{thm: CR tight at time 1}, at time 1, $CR$ is tight. So,
    \[OPT = \alpha_f(1) = \int_{t=0}^{1}{\frac{t+1-a+a(1-t)}{1}f_t \, dt} + \int_{t=1}^{\infty}{f_t \, dt}.\]
    Let $E \coloneqq \int_{t=0}^{1}{tf_t \, dt}$, then the above equality becomes $E=\frac{OPT-1}{1-a}$. Thus $aE = \frac{OPT-1}{\frac{1}{a}-1} \leq \delta$. On the other hand, since we are in world 2, by Lemma~\ref{lem: monotonicity of zero interval }, $c-aE < 0$. So, $c < aE \leq \delta$. By Theorem~\ref{thm: bad interval suffix has exactly delta mass}, in world 2 we must have $\delta$ mass in $(L_b,\infty]$. Thus $c \geq \delta$. This leads to a contradiction.
\end{proof}

\begin{restatable}{lemma}{BadIntervalIsOfLengthLessThanOne}
    \label{lem: bad interval has length less than 1}
    Any bad interval of $x \in (I_3,L_b]$ is of length less than 1.
\end{restatable}

\begin{proof}
    Recall that $L_b = \frac{\gamma-1}{1-a\gamma}$ and $\gamma < \frac{1}{a}$. The bad interval of $x$ is $(\frac{(\gamma-1)a}{1-a}x+\gamma-1,x]$ which has length $\frac{1-a\gamma}{1-a}x-(\gamma-1)<\frac{1-a\gamma}{1-a}L_b-(\gamma-1)=(\frac{1}{1-a}-1)(\gamma-1)<\frac{a}{1-a}(\frac{1}{a}-1)=1.$
\end{proof}

\OurContributionFOURGreedyStructure*

\begin{proof}
    Let $f_t$ be the greedy solution. By Theorem~\ref{thm: in any opt solution something is tight in [0,1]}, the first statement holds. By Lemma~\ref{lem: exponential function to make CR tight}, the second statement holds.
    
    To prove the third statement in the theorem, we only need to show that if at some point before $L_b$ the greedy solution has to be mass tight, then it will keep mass tight until $L_b$. We prove this point in the discrete setting for convenience, while the result that mass tight at every point in $(p,L_b]$ can be generalized to a continuous setting. Let $T$ be the first point where the mass at time~$T$ is larger than the mass at time $T+\tau$. Then there are two cases: either $CR$ is tight up to time $T$ and mass is tight at time $T+\tau$ with $f_T>f_{T+\tau}$, or $CR$ is tight up to time $T-\tau$ and mass is tight at time $T$ with $f_{T-\tau} \leq f_T < (1+\tau)f_{T-\tau}$. Let us assume we are in the first case without loss of generality (proof for the second case is similar). Let $T_1$ be the first point where $T$ is not in the bad interval of $T_1$ but $T+\tau$ is. We call the bad interval of $T_1$ the first interval. Similarly, let $T_2$ be the first point where $T_1$ is not in the bad interval of $T_2$ but $T_1+\tau$ is, and we call the bad interval of $T_2$ the second interval. We define the $k^{th}$ interval in the same way in $[T,L_b]$ for each $k$, where the last interval is chopped at $L_b$ as a sub-interval of $(1,L_b]$.

    Let $\bar{f}$ be the mass distribution where $\bar{f_t}=f_t$ for each point $t \leq T$, and $\bar{f_t}$ has the amount of mass to make each point $t>T$ mass tight until $L_b$. We will show that each point $T<x \leq L_b$ in $\bar{f}$ satisfies the $CR$ constraint, i.e. $\alpha_{\bar{f}}(x) \leq OPT$. Then by the fact that $\alpha_f(x)$ depends only on $f_t$ where $t \leq x$, we know $\bar{f_t}=f_t$ for each point $t$ from left to right, hence they are equal the whole way. Thus $f_t$ has the required structure and we are done.

    By the definition of the $k^{th}$ interval, we know that the set of points inside the $k^{th}$ interval coincides with the set of points in some bad interval of some point, hence by Lemma~\ref{lem: bad interval has length less than 1}, the convex hull of points in the $k^{th}$ interval is of length less than 1. Let us explore what is the mass distribution within the $k^{th}$ interval. For the first interval, the first point is the one to make mass tight. Since we discretize points with time step $\tau$, for each point $x$ in the first interval, the bad interval of $x$ contains some previous discrete mass values. Note that the bad interval grows when $x$ is larger, and the left bound of the bad interval moves slower than the speed that $x$ moves. Hence, each time we have either $\bar{f_x}=0$ (where the bad intervals of $x$ and $x-\tau$ contain the same set of discrete points), or $\bar{f_x} = \bar{f_t}$ where $t<x$ is the point that just escaped the bad interval of $x$. Thus, by the definition of $T_1$, $\bar{f}_{T_1}=\bar{f_T}$. Thus, in the first interval (as well as in each defined interval), there is a bunch of zero mass and some repetitions of previous mass values. It turns out that $\bar{f_T}=\bar{f}_{T_1}=\bar{f}_{T_2}=\dots$ has the largest value of mass among $\bar{f_t}$ where $1<t \leq L_b$.

    Now, we show that the $CR$ of any point in the $k^{th}$ interval is less than the $CR$ of the last point in the $(k-1)^{th}$ interval $T_{k-1}$ for each $k=1,2,\dots$, where the last point in the $0^{th}$ interval is $T$. For each step $k=1,2,\dots$ we do the following: assume that there is a distribution~$\hat{f}$ satisfying $\hat{f}_t = \bar{f}_t$ at any point before the $k^{th}$ interval and $\hat{f}_t$ keeps $CR$ constant in the $k^{th}$ interval, with the constant value $\alpha_{\bar{f}}(T_{k-1})$. Since we want the $CR$ to be flat in the $k^{th}$ interval, the mass distribution of $\hat{f}_t$ must be an increasing function in the $k^{th}$ interval. Thus, $\hat{f}_t > \hat{f}_{T_{k-1}} = \bar{f}_{T_{k-1}} \geq \bar{f}_{t}$ for any $t$ in the $k^{th}$ interval. In other words, $\hat{f}_t$ pointwise dominates~$\bar{f}_t$ in the $k^{th}$ interval. We know that when $x \geq 1$, $\alpha_f(x) = \frac{x}{1-a+ax}+ \sum_{t=0}^{x} \frac{(1-a)(1+t-x)}{1-a+ax}f_t$. It implies that decreasing the value of $f_t$ when $x-1 < t \leq x$ only benefits $\alpha_f(x)$. Hence, for point $x$ in the $k^{th}$ interval, since any other point is within a distance of less than 1 to point $x$, we have $\alpha_{\bar{f}}(x) < \alpha_{\hat{f}}(x) = \alpha_{\hat{f}}(T_{k-1}) = \alpha_{\bar{f}}(T_{k-1})$. 

    Thus, we know that $\alpha_{\bar{f}}(x) < OPT$ for any $T < x \leq L_b$, which means that after the first mass tight point in $f_t$, at each point we need more mass to make $CR$ tight than to make mass tight. This completes the proof for the third statement.

    Finally, the last statement holds by the definition and feasibility of the greedy solution.
\end{proof}

\OurContributionFiveAllDelta*

\begin{proof}
    In the small $\delta$ regime, by Theorem~\ref{thm: bad interval suffix has exactly delta mass} there is exactly $\delta$ mass over $(L_b,\infty]$. Hence, we can put all the $\delta$ mass in the bad interval suffix to $\infty$ and still get a feasible optimal solution without violating the $CR$ constraint at $\infty$. Thus in both small and large $\delta$ regime, we have an optimal solution with support $[0,L_b] \cup \{\infty\}$.
\end{proof}


\section{Algorithms}
\label{sec:algorithms}

In this section, we will design algorithms that give optimal solution(s), provided that we have one tail constraint~$(\gamma,\delta)$. We will give two algorithms, a binary search algorithm and an LP algorithm. Recall that in Section~\ref{subsec: optimal solution by greedy algorithm}, the greedy algorithm can be used to analyze the structure of optimal solutions, but they cannot directly construct a solution because the optimal competitive ratio remains unknown, and hence the exact amount of mass required to achieve $CR$ tight is also unknown. However, based on the greedy algorithm, we show that there is a binary search algorithm that takes in an accuracy parameter $\epsilon>0$ and a guess of OPT in each step (Theorem~\ref{our contribution:thm 6, binary search}). On the other hand, as a consequence of Theorem~\ref{our contribution:thm 5, closely related greedy, all delta}, one can also design an LP algorithm (Theorem~\ref{our contribution:thm 7, LP}). We will discuss the trade-off between these algorithms.

\subsection{Binary Search Algorithm}
\label{subsec: binary search algorithm}

In Section~\ref{subsec: optimal solution by greedy algorithm}, we have a greedy algorithm which needs to know the information of OPT. A natural choice is to guess OPT and do a binary search. This subsection will give a way of constructing an optimal solution to the problem by guessing OPT. A mathematical program is needed to prove the correctness of our construction. Since we are going to use mathematical programming, we will consider discrete time points: let $\tau>0$ be a sufficiently small time step so that for any critical point $B$, we have $B/\tau$ as an integer, where the set of critical points consists of all boundary points of $BI_1, \dots, BI_5$, and time $1$ and $L_b$. Thus, we will find the mass distribution over discrete time $1 \cdot \tau, 2 \cdot \tau, 3 \cdot \tau, \cdots$. We relax the constraint that $\sum_{t=\tau}^{\infty} f_t=1$, hence $f_t$ is not necessarily a probability density function (PDF).

Let $T$ be a guess of OPT. First we want to find some similar zero interval properties of a mass distribution (may not sum up to 1) under a guess of value $T$, given the mass constraint and the ``competitive ratio constraint'' (we call it $T-$competitive ratio constraint or in short~$T$-$CR$, meaning that the $CR$ of every point does not exceed our guess $T$). We will show that the guess of $T$ will not change the monotonic property over the zero interval after 1 and asymptotic property of the zero suffix.

\begin{lemma}
    \label{lem: T-monotonicity over zero interval}
    For any mass distribution $f$ with a zero interval in $[1,\infty)$, given the mass constraint and the $T$-competitive ratio constraint, the $T$-competitive ratio is a monotone function of $x$ on the zero interval.
\end{lemma}

\begin{proof}
Let $f$ be a mass distribution not necessarily sum up to 1.  Let $I \coloneqq [T_1,T_2]$ be a zero interval of $f$ after point 1. Also denote $P \coloneqq [0,T_1)$ and $S \coloneqq (T_2,\infty]$, and the $T$-$CR$ as $\alpha_f^T(x)$. Thus for any $x \in [T_1,T_2]$,
\begin{align*}
\alpha_f^T(x) =& \int_P { \frac{t+1-a+a(x-t)}{1-a+ax}f_t \, dt} + \int_S {\frac{x}{1-a+ax}f_t \, dt}\\
=& \int_P {f_t \, dt} + \frac{(1-a)}{1-a+ax} \int_P { t f_t \, dt} + \frac{x}{1-a+ax} \int_S {f_t \,dt}.
\end{align*}

For a given $f$, denote $c_1 \coloneqq \int_P {f_t \,dt}$, $E \coloneqq \int_P tf_t \,dt$, and $c_2 \coloneqq \int_S {f_t \,dt}$. Then the above competitive ratio function becomes
\[\alpha_f^T(x) = c_1 + \frac{1-a}{1-a+ax}E+\frac{x}{1-a+ax}c_2.\]
Taking the derivative of $\alpha_f^T(x)$ implies that ${\alpha_f^T}^{\prime}(x)=\frac{(1-a)(c_2-aE)}{(1-a+ax)^2}$, which is positive (negative) when the fixed threshold $c_2-aE$ is positive (negative), and is $0$ when $c_2-aE=0$. Thus, $\alpha_f^T(x)$ is monotone on the zero interval $[T_1,T_2]$.

We state that with a similar proof of Lemma~\ref{lem: asymptotic property holds for any feasible distribution}, one can also find that the $T$-$CR$ is strictly decreasing over the zero interval in $[0,1]$ with a guess of $T$. Note that the asymptotic property also holds in this setting. Let $[T_0,\infty)$ be a zero suffix. Any point $x \in [T_0,\infty)$ has a $T-$competitive ratio $\alpha_f^T(x) =  \int_{t=0}^{T_0} { f_t \, dt} + \frac{1-a}{1-a+ax} \int_{t=0}^{T_0} { t f_t \, dt} + \frac{x}{1-a+ax} f_{\infty}$. When $x \to \infty$, $\alpha_f^T(x)$ converges to $f_{<\infty}+0+\frac{1}{a}f_\infty$, where $f_{<\infty}$ is the total mass over $[0,\infty)$. On the other hand, at $x=\infty$, we have $\alpha_f^T(x) =  \int_{<\infty} { \frac{t+1-a+a(x-t)}{1-a+ax}f_t \, dt} + \frac{x}{1-a+ax}f_\infty$ which becomes $f_{<\infty}+\frac{1}{a}f_\infty$. Thus our asymptotic property holds.
\end{proof}

Now, we can construct an algorithm given a guess of $T$. We give $ALG(a, \gamma, \delta, T)$  in Algorithm~\ref{alg:ALG_subroutine}.

\begin{algorithm}[ht]
\caption{ALG\((a,\gamma,\delta,T)\), a subroutine for the binary search}\label{alg:ALG_subroutine}
\begin{algorithmic}[1]
  \Require $a,\gamma,\delta,T$
  \Statex \textbf{Notation:}
    \(\;L_3\) is the left bound for \(BI_3\),  
    \(\;I_3\) is the right bound for \(BI_3\),  
    \(\;L_5\) is the left bound for \(BI_5\).
  \State \(f_\tau \gets \min\bigl(\delta,\;\tfrac{T-1}{1-a}\,\tau\bigr)\)
  \For{\(x = 2\tau,3\tau,\dots\)}
    \If{\(\sum_{t<x}f_t \ge 1\)}
      \State \Return \(\{f_t\}\)
    \EndIf

    \If{\(x \le L_3\)}
      \State
      \(f_x \gets
        \min\Bigl(
          \delta - \sum_{t\in I_{\gamma}(x)\setminus\{x\}}f_t,\;
          \tfrac{x}{1-a}\bigl(T-1-\sum_{t<x}\tfrac{(1-a)(1+t-x)}{x}f_t\bigr)
        \Bigr)\)

    \ElsIf{\(L_3 + \tau \le x \leq 1\)}
      \State
      \(f_x \gets
        \tfrac{x}{1-a}\bigl(T-1-\sum_{t<x}\tfrac{(1-a)(1+t-x)}{x}f_t\bigr)\)

    \ElsIf{\(x = 1+\tau\)}
      \State
      \(f_x \gets
        \tfrac{1-a+ax}{1-a}\bigl[T-\tfrac{x}{1-a+ax}
        -\sum_{t<x}\tfrac{(1-a)(1+t-x)}{1-a+ax}f_t\bigr]\)
      \If{\(f_{1+\tau}=0\)}
        \State \(f_\infty \gets a\bigl(T-\sum_{t=\tau}^{1}f_t\bigr)\)
        \State \Return \(\{f_t\}\)
      \EndIf

    \ElsIf{\(1+2\tau \le x \le I_3\)}
      \State
      \(f_x \gets
        \tfrac{1-a+ax}{1-a}\bigl[T-\tfrac{x}{1-a+ax}
        -\sum_{t<x}\tfrac{(1-a)(1+t-x)}{1-a+ax}f_t\bigr]\)

    \ElsIf{\(I_3+\tau \le x \le L_5\)}
      \State
      \(f_x \gets
        \min\Bigl(
          \delta - \sum_{t\in I_{\gamma}(x)\setminus\{x\}}f_t,\;
          \tfrac{1-a+ax}{1-a}\bigl[T-\tfrac{x}{1-a+ax}
          -\sum_{t<x}\tfrac{(1-a)(1+t-x)}{1-a+ax}f_t\bigr]
        \Bigr)\)

    \ElsIf{\(x \ge L_5+\tau\)}
      \State
      \(f_x \gets
        \min\Bigl(
          \delta - \sum_{t=L_b+\tau}^{\,x-\tau}f_t,\;
          \tfrac{1-a+ax}{1-a}\bigl[T-\tfrac{x}{1-a+ax}
          -\sum_{t<x}\tfrac{(1-a)(1+t-x)}{1-a+ax}f_t\bigr]
        \Bigr)\)
      \If{\(f_x = 0\)}
        \State \Return \(\{f_t\}\)
      \EndIf

    \EndIf
  \EndFor
\end{algorithmic}
\end{algorithm}

\begin{lemma}
    \label{lem: algorithm terminates}
    Algorithm~\ref{alg:ALG_subroutine} will terminate.
\end{lemma}

\begin{proof}
    Clearly, when $\sum_{t<x}f_t \geq 1$ for some point $x$, the algorithm will terminate. When $\sum_{t<x}f_t <1$ in the whole domain, if $f_{1+\tau}=0$, then we know that in this case a zero mass can make $T$-$CR$ tight at time $1+\tau$, the algorithm will terminate. If not, we know that some nonzero mass is needed to make $T$-$CR$ tight at time $1+\tau$, hence the $T$-$CR$ is decreasing over the zero interval after time 1. In this case, we want to show that the algorithm terminates. By our construction for $x \geq L_5+\tau$, if the amount of mass at time $x$ to make mass tight (the first term) is smaller than the one to make $T$-$CR$ tight (the second term), then we know that $\sum_{t=L_b+\tau}^{x} f_t = \delta$, and $f_{x+\tau}=0$. If the second term is smaller, by the monotonicity of $T$-$CR$ over a zero interval, we know that more mass is needed at the next point to make $T$-$CR$ tight compared to the current point. In other words, it is impossible that the mass distribution goes on infinitely long with less and less mass to make the total mass converges to some number. Thus, mass tight will dominate the term finally, and hence the above algorithm can terminate.
\end{proof}

Given Algorithm~\ref{alg:ALG_subroutine} as a subroutine, we provide the binary search algorithm\footnote{Since the greedy algorithm gives optimal solution for every possible tail constraint $(\gamma,\delta)$, the binary search algorithm is able to compute the optimal distribution even for a collection of tail bounds $\{(\gamma_i, \delta_i)\}_{i=1}^k$, by taking the minimum amount of mass one can add each time into account in the subroutine algorithm, among all mass constraints.}, which takes in an accuracy parameter $\epsilon > 0$. Given $2-a \leq \gamma <\frac{1}{a}$ and $0 \leq \delta \leq 1$. At each step we maintain an interval $[l,u]$ such that $T \in [l,u]$. Then, we give Algorithm~\ref{alg:binary_search_alg}.

\begin{algorithm}[ht]
\caption{BinarySearch\((a,\gamma,\delta,\epsilon)\)}\label{alg:binary_search_alg}
\begin{algorithmic}[1]
  \Require $a,\gamma,\delta,\epsilon$
  \State $l \gets \dfrac{e}{e-1+a},\quad u \gets 2 - a$
  \While{$u - l > \epsilon$}
    \State $T \gets \dfrac{l + u}{2}$
    \State $f \gets \mathrm{ALG}(a,\gamma,\delta,T)$
    \If{$\sum_t f_t < 1$}
      \State $l \gets T$
    \Else
      \State $u \gets T$
    \EndIf
  \EndWhile
  \State $x^* \gets \min\{\,x : \sum_{t\le x}f_t \ge 1\}$
  \State \Return
  $\displaystyle
    f'_x =
    \begin{cases}
      f_x, & x < x^*,\\
      1 - \sum_{t < x^*}f_t, & x = x^*,\\
      0, & x > x^*.
    \end{cases}$
\end{algorithmic}
\end{algorithm}

We show that Algorithm~\ref{alg:binary_search_alg} is the required algorithm in Theorem~\ref{our contribution:thm 6, binary search}.

\OurContributionSixBinarySearch*

For the remaining part of Section~\ref{subsec: binary search algorithm}, we prove Theorem~\ref{our contribution:thm 6, binary search}. Let us start from the construction $ALG(a,\gamma, \delta, T)$. Let $x^*$ be the first point with $\sum_{t \leq x^*}f_t \geq 1$.
We will show that Algorithm~\ref{alg:ALG_subroutine} solves the following mathematical program (MP). This helps us prove the correctness of search procedure.

\[
\begin{aligned}
 \max  & \  \sum_{t=\tau}^{\infty} f_t \\
\text{s.t.} & \quad \sum_{t \leq x} \frac{(1-a)(1+t-x)}{x} f_t + 1 \leq T \quad &&\text{for } x=\tau, 2\tau, \dots, 1 \\
& \quad \sum_{t \leq x} \frac{t+1-a+a(x-t)}{1-a+ax} f_t + \frac{x}{1-a+ax} \max(0,1-\sum_{t \leq x} f_t) \leq T \quad &&\text{for } x=1, 1+\tau, \dots \\
& \quad \sum_{t \in I_{\gamma}(x)} f_t \leq \delta  \quad &&\text{for } x=\tau, 2\tau, \dots, L_5 + \tau  \\
& \quad f_t \geq 0 \quad &&\text{for } x=\tau, 2\tau, \dots
\end{aligned}
\]

Denote the above mathematical program as $MP(T)$, where $T$ is our guess of OPT. The first and second lines of constraints represent the $T$-$CR$ constraints, and the third line of constraints represents the mass constraints. For brevity, we call these three lines of constraints nontrivial constraints. Now, let $f$ be any optimal solution with the longest tight prefix among all optimal solutions to the problem $MP(T)$. The next lemma shows the tightness of nontrivial constraints in $[0,min(x^*,L_b)]$.

\begin{lemma}
    \label{lem: tightness to Lb for MP}
    Let $f$ be any optimal solution with the longest tight prefix among all optimal solutions to $MP(T)$. Then, for each $x=\tau,2\tau,\dots,min(x^*,L_b)$, some nontrivial constraint in $MP(T)$ must be tight.
\end{lemma}

\begin{proof}
    
    Suppose the lemma does not hold. Let $B$ be the set of $x \in \{\tau, \dots, \infty \}$ so that neither nontrivial constraint is tight at $x$, then $B \cap [0,min(x^*,L_b)] \neq \emptyset$.  Note that $\forall x \in B$ we have both $\sum_{t \in I_{\gamma}(x)} f_t < \delta$, and
    \[\begin{cases}
    \sum_{t \leq x} \frac{(1-a)(1+t-x)}{x} f_t + 1 < T \quad & \text{when} \ x \in B \cap [0,1] \\
    \sum_{t \leq x} \frac{t+1-a+a(x-t)}{1-a+ax} f_t + \frac{x}{1-a+ax} max(0,1-\sum_{t \leq x} f_t) < T \quad & \text{when} \ x \in B \cap [1, \infty] .
    \end{cases} \label{eq:fakeCR} \tag{\(\dagger\)}\]
    Consider the structure of set $B$. If $B$ is a suffix of the form $B= \{x^\prime, x^\prime+\tau, x^\prime+2\tau, \dots \}$, then $B \cap [0,min(x^*,L_b)] \neq \emptyset$ implies that $x^\prime \leq min(x^*,L_b)$. Thus we consider
    \[f_t^\prime = \begin{cases}
    f_{x^\prime} + \epsilon \quad & \text{if} \ t = x^\prime \\
    f_{x^\prime} \quad & \text{otherwise}.
    \end{cases}\]
    Since $x^\prime \leq min(x^*,L_b) \leq L_b$, if $x^\prime$ lies in some bad interval, then that must be the bad interval of some elements in $B$, which are all slack. For the \eqref{eq:fakeCR} constraints, adding $\epsilon$ mass at $x^\prime$ only changes the LHS of \eqref{eq:fakeCR} for $x > x^\prime$. But such $x$ only lies in $B$ which are all slack. It is safe to add mass and will not violate \eqref{eq:fakeCR} constraints for all finite $x$ values.
    On the other hand, when $x \to \infty$, LHS of the second inequality of \eqref{eq:fakeCR} cannot converge to $T$. This is because, if the max term takes $1-\sum_{t \leq x} f_t$, adding an $\epsilon$ amount of mass at some finite point $x^\prime$ decreases LHS at $x \to \infty$. If the max term takes 0, then we have $\sum_{t \leq x}f_t \geq 1$ and thus halt after assigning mass at $x^*$. It turns out that $f_t$ has a zero suffix. This implies that we have a decreasing LHS in $f$ and hence cannot converge to $T$. Thus we are safe to add mass at $x^\prime$. However, this gives us a larger total amount of mass, contradicting to the optimality of $f$.
    
    If $B$ is not a suffix, let $t_1 = min_{x \in B} x$ and $t_2=min_{x \not\in B, x>t_1} x$. Since $B \cap [0,min(x^*,L_b)] \neq \emptyset$, we know that $t_1 \leq min(x^*,L_b)$. If $f_{t_2} \neq 0$, let us move an $\epsilon$ amount of mass (sufficiently small) from $t_2$ to $t_1$, that is,
    \[f_t^{*} = \begin{cases}
    f_{t_1} + \epsilon \quad & \text{if} \ t = t_1 \\
    f_{t_2} - \epsilon \quad & \text{if} \ t = t_2 \\
    f_{t} \quad & \text{otherwise}.
    \end{cases}\]
    If we are allowed to do this (not violating any constraint), the objective value is preserved as the sum of $f_t$ unchanged, so $f_t$ is optimal will imply that $f_t^*$ attains the optimal objective value. For any $x<t_1$, since $t_1 \leq min(x^*,L_b) \leq L_b$, its bad interval $I_\gamma(x)$ (if non-empty) is of the form $(\cdot,x]$ (call it form 1). Hence moving some mass from $t_2$ to $t_1$ will not affect~$I_\gamma(x)$. For the \eqref{eq:fakeCR} constraints, since its LHS only depends on $f_t$ where $t \leq x$, thus there is no change; For any $x \in \{t_1, t_1+\tau, \dots, t_2-\tau\}$, by the definition of $t_2$, all such $x$ are in $B$ and hence constraints associated with such $x$ are slack in $f$. Thus, we can choose $\epsilon$ so that constraints~\eqref{eq:fakeCR} at $x$ keeps valid. Again, since $t_1 \leq L_b$, we can also choose $\epsilon$ so that the mass constraint is valid no matter where $t_2$ is; For any $x \geq t_2$, if $t_1 \in I_\gamma(x)$, then $t_2$ must be in $I_\gamma(x)$ and hence the mass constraint at $x$ is valid. For constraints \eqref{eq:fakeCR}, if the max term takes $1-\sum_{t \leq x} f_t$ at $x$, then both two lines in \eqref{eq:fakeCR} are indicating the $T$-$CR$, i.e. the `competitive ratio constraints' with a guess of OPT $T$. In this case, moving mass from $t_2$ to $t_1$ decreases the LHS at $x$. If the max term takes 0, LHS of the second line of constraints in \eqref{eq:fakeCR} becomes $\sum_{t \leq x} \frac{t+1-a+a(x-t)}{1-a+ax} f_t$. When $x \geq t_2$, numerator changes will be $(1-a)t_1\epsilon-(1-a)t_2\epsilon$ which is negative. Thus at $x$ the LHS decreases, and \eqref{eq:fakeCR} remains to be true. Note that this makes both constraints at $t_2$ decrease. Thus we can repeat the procedure until either $t_1$ becomes something tight in some optimal solution (which is a contradiction to the longest tightness assumption of $f$), or we find another optimal solution with the same longest tight prefix but a nontight suffix (and apply the previous case).

    Suppose that $f_{t_2} = 0$. If some constraint in \eqref{eq:fakeCR} is tight at $t_2$ (does not hold), by the monotonicity property given the guess of $T$ (Lemma~\ref{lem: T-monotonicity over zero interval}), a local first tight point must have nonzero mass, so we have a contradiction. Thus, mass is tight at $t_2$. Since $t_2$ is the first something tight point after $t_1$ and $f_{t_2}=0$, $t_2$ must be $L_5$. Then $f_{t_2} = 0$ and $t_2-\tau \in B$ indicate that there exists $t_3 > t_2$ with $f_{t_3} \neq 0$ (otherwise $t_2$ cannot be mass tight). Without loss of generality, $t_3$ is the first nonzero mass point after $t_2$. We move some tiny mass from~$t_3$ to $t_1$. This is a feasible solution since $t_1 \leq min(x^*,L_b) \leq L_b$ means that $t_1$ is only in the bad interval of some $x<t_2=L_5$ whose mass constraint is not tight. Clearly, \eqref{eq:fakeCR} holds at any point in $[t_1,t_2]$ and also holds at any point in $[t_2,t_3-\tau]$ which is a zero interval, by the monotonicity property over the zero interval given the guess of $T$ (Lemma~\ref{lem: T-monotonicity over zero interval}). The fact that~\eqref{eq:fakeCR} holds at any point outside $[t_1,t_3-\tau]$ is trivial. Thus we go back to the previous case.

    We finish the proof as required.
\end{proof}

This lemma implies the following theorem that our construction of $f$ gives an optimal solution to the mathematical program.

\begin{lemma}
    \label{lem: ALG is optimal for MP}
    Let $f$ be the output of $ALG(a,\gamma, \delta, T)$, and let $f^*$ be an optimal solution to $MP(T)$ with the longest tight prefix. Then $f$ gives the same optimal distribution to $f^*$ from 0 to $min(x^*,L_b)$.
\end{lemma}

\begin{proof}
    First it is easy to note that the solution constructed by our algorithm $ALG(a,\gamma, \delta, T)$ is feasible for the $MP(T)$ (any possibly chopping off solution of a feasible solution is still feasible). It suffices to show by induction that $f_x=f_x^*$ for $x \leq min(x^*,L_b)$.

    For the base case, consider $x=\tau$ and we must have $f_\tau^* \leq f_\tau = min(\delta, \frac{T-1}{1-a} \tau)$ since $f^*$ is feasible. If $f_\tau^* < f_\tau$, then at time $\tau$ neither is tight at $f^*$. This violates Lemma~\ref{lem: tightness to Lb for MP}.

    For $x>\tau$, assume for induction that $f_t^*=f_t$ for all $t<x$. First, $f_x^* \leq f_x$. This is because feasibility of $f_x^*$ to $MP(T)$ leads to valid \eqref{eq:fakeCR} constraints when the summation is over $t \leq x$, which consists of two terms $t<x$ and at $x$. By the induction hypothesis, $f_x^*=f_x$ for the summation over $t<x$. Thus we find the required upper bound for the $f_x^*$ term. Again, if $f_x^*<f_x$, it violates Lemma~\ref{lem: tightness to Lb for MP} at time $x$. Thus $f_x^*=f_x$ when $x\leq min(x^*,L_b)$. This completes our proof.
\end{proof}

Now, we want to prove that the solution returned by our algorithm has more than one unit of mass if and only if our guess is larger than the true OPT value. This allows us to complete the search procedure.

\begin{lemma}
    \label{lem: sum>=1 iff T>=OPT}
    Let $f$ be the output of $ALG(a,\gamma, \delta, T)$. Then $\sum_t f_t \geq 1 \iff T \geq OPT$. 
\end{lemma}

\begin{proof}
    For the forward direction, assume that $\sum_t f_t \geq 1$. We consider 
    \[f_x^{\prime} = \begin{cases}
    f_x \quad & \text{if} \ \sum_{t \leq x} f_t < 1 \\
    1 - \sum_{t < x} f_t \quad & \text{if} \ \sum_{t \leq x-\tau} f_t < 1 \ \text{but} \ \sum_{t \leq x} f_t \geq 1 \\
    0 \quad & \text{otherwise}. 
    \end{cases}\]
    Then $f^\prime$ is a probability distribution. Let $x^\prime$ be the point such that $\sum_{t=\tau}^{x^\prime} f_t^\prime =1$. We want to show that $f^\prime$ satisfies all tail constraints, and $f^\prime$ has expected competitive ratio at most $T$, implying that $OPT \leq T$.
    Since we get $f^\prime$ by chopping from $f$ which satisfies tail constraints, $f^\prime$ also satisfies all tail constraints. Now let us consider the ratio at each point $x$ in $f^\prime$. Let the LHS of the $T$-$CR$ constraint to be $C(x)$ for $f$, and $C^\prime(x)$ for $f^\prime$. For any $x$ such that $\sum_{t=\tau}^{x}f_t \leq 1$, we have $C^\prime(x) = C(x) \leq T$, where the equality is because the max term takes the nonzero term so that $f^\prime=f$, and the inequality holds since $f$ is feasible for $MP(T)$. For any $x$ such that $\sum_{t=\tau}^{x}f_t^\prime > 1$, at these $x$, we have $f_x^\prime=0$. Then
    \begin{align*}
    C^\prime(x) =& \sum_{t=\tau}^{x^\prime} \frac{t+1-a+a(x-t)}{1-a+ax} f_t^\prime + \sum_{t=x^\prime+dt}^{x} \frac{t+1-a+a(x-t)}{1-a+ax} f_t^\prime\\
    \leq& \sum_{t=\tau}^{x^\prime} \frac{t+1-a+a(x-t)}{1-a+ax}f_t+\sum_{t=x^\prime+dt}^{x}\frac{t+1-a+a(x-t)}{1-a+ax}f_t = C(x) \leq T
    \end{align*}
    where the last equality holds since the max term takes 0, and the last inequality holds since~$f$ is feasible.

    For the backward direction, assume that $T \geq OPT$. If $x^*=min(x^*,L_b)$, we place at least one unit of mass when reaching time $x^*$ by its definition. Thus, $\sum_t f_t \geq 1$. Otherwise,~$L_b$ is smaller. In this case, note that the real optimal solution to the true problem is a feasible solution to $MP(T)$. It turns out that any optimal solution to $MP(T)$ has at least one unit of mass. By the mass constraint in $MP(T)$, there is at most $\delta$ mass in $(L_b,\infty]$. Thus, any optimal solution to $MP(T)$ has at least $1-\delta$ mass in $[0,L_b]$. By Lemma~\ref{lem: ALG is optimal for MP}, $f$ gives an optimal distribution to $MP(T)$ in $[0,min(x^*,L_b)]$ (which is now $[0,L_b]$) and thus has at least $1-\delta$ mass in $[0,L_b]$. To show that $\sum_t f_t \geq 1$, it suffices to show that we can put $\delta$ mass in $(L_b,\infty]$. Note that in our $T-$competitive ratio setting, $T$-$CR$ is tight at time 1. By Lemma~\ref{lem: T-monotonicity over zero interval}, if there is a zero interval after time 1, we cannot have an increasing $T$-$CR$. Clearly, if it gives us a strictly decreasing $T$-$CR$ over the zero interval, so does any zero interval after time 1 since $c_2-aE$ in Lemma~\ref{lem: T-monotonicity over zero interval} decreases while increasing $x$ in $f$. In this case, we have enough room to place $\delta$ mass in the bad interval suffix as long as satisfying the $T$-$CR$ constraint. Thus $f$ gives an optimal solution to $MP(T)$ and $\sum_t f_t \geq 1$. If the zero interval gives us a flat ratio, there must be a zero suffix $(1,\infty)$ and by the asymptotic property at $\infty$ we must have $T$-$CR$ tight. Hence there is only one way of placing mass, i.e. tightness holds at every point and hence our algorithm computes the optimal solution to the $MP(T)$. Because of the fact that any optimal solution to $MP(T)$ has at least one unit of mass, we can conclude that $\sum_t f_t \geq 1$.
\end{proof}

Finally, we go back to prove Theorem~\ref{our contribution:thm 6, binary search}.
 
\begin{proof}[Proof of Theorem~\ref{our contribution:thm 6, binary search}]
  By Lemma~\ref{lem: sum>=1 iff T>=OPT}, we have $\sum_t f_t \geq 1 \iff T \geq OPT$. Therefore, by making a guess $T$ each round, our binary search algorithm will push our guess closer to the true optimal $CR$. Since the size of the initial search range is $2-a-\frac{e}{e-1+a}$ and we reduce the length by a factor of 2 after each query, the binary procedure queries our algorithm $O(\log(\frac{1}{\epsilon}))$ times, resulting in the last search space $[l,u]$ with $u-l \leq \epsilon$. Finally, we truncate the construction $f$ to make it feasible. This finishes the proof of this theorem.
\end{proof}

We know that to keep the competitive ratio tight after 1, we need exponentially increasing mass. This is the same for our $T-CR$. The binary procedure only queries $O(log(\frac{2-a-\frac{e}{e-1+a}}{\epsilon}))$ times. Because of the log, only few queries needed ($log(\frac{2-a-\frac{e}{e-1+a}}{\epsilon}) \approx 12$) when we set small $\epsilon=10^{-6}$ and $a=0.5$. Thus, our construction is reasonably fast, and so is the binary search procedure. However, by the nature of the binary search procedure, we cannot get the exact optimal $CR$ to the problem. In Section~\ref{subsec:LP algorithm}, we argue that one can use another approach of linear programming to get the exact value, but significantly slower.

\subsection{Linear programming algorithm}
\label{subsec:LP algorithm}

To find an exact optimal solution to our problem, one can take advantage of the linear program. In Section~\ref{subsec: general structure of optimal solutions}, we argued that there always exists an optimal solution with a fixed, finite support $[0,L_b] \cup \{\infty\}$ (Theorem~\ref{our contribution:thm 5, closely related greedy, all delta}). In this section, we will introduce the LP algorithm by given the LP formulation and solving the~LP. The finite support property of Theorem~\ref{our contribution:thm 5, closely related greedy, all delta} allows us to set a linear program (LP) to solve the problem, namely we minimize the competitive ratio~$\lambda$ without violating $CR$ and mass constraints at any point in $[0,L_b]$ and the mass over the support is sum to 1.

Let us give the LP formulation based on the result that there always exists a solution with a fixed finite support $[0,L_b] \cup \{\infty\}$ to the problem.

\begin{equation}
\tag{LP}\label{eq:lp}
\begin{aligned}
  \min\; & \lambda \\
  \text{s.t.}\; 
    & \sum_{t \le x} \frac{t+1-a+a(x-t)}{x} f_t
      + \sum_{x<t\le L_b} f_t
      + f_\infty \le \lambda
      &&\text{for }x=\tau,2\tau,\dots,1,\\
    & \sum_{t \le x} \frac{t+1-a+a(x-t)}{1-a+ax} f_t
      + \sum_{x<t\le L_b} \frac{x}{1-a+ax}f_t
      + \frac{x}{1-a+ax}f_\infty \le \lambda
      &&\text{for }x=1,1+\tau,\dots,L_b,\\
    & \sum_{t\in I_\gamma(x)} f_t \le \delta 
      &&\text{for }x=\tau,2\tau,\dots,L_b,\\
    & f_t \ge 0 
      &&\text{for }t=\tau,2\tau,\dots,L_b,\\
    & \sum_{t \leq L_b} f_t + f_\infty = 1.
\end{aligned}
\end{equation}

We have the following LP algorithm (Algorithm~\ref{alg:solveLP}) and Theorem~\ref{our contribution:thm 7, LP}.

\begin{algorithm}[ht]
\caption{A linear program based algorithm}\label{alg:solveLP}
Solve the linear program \eqref{eq:lp}.\;
\end{algorithm}

\OurContributionSevenLP*

\begin{proof}
    The correctness follows from the $LP$ formulation and Theorem~\ref{our contribution:thm 5, closely related greedy, all delta}.
\end{proof}

By using some LP solver, one can get the exact optimal competitive ratio and the corresponding optimal distribution. A major downside of this LP algorithm is that it is much slower than implementing the combinatorial binary search procedure, especially when the size of the problem is sufficiently large. The running time of the LP algorithm depends on general LP solvers and scales with the number of variables and constraints derived from the problem. In our LP formulation, one can observe that there are $\Theta(L_b/\tau)$ variables and $\Theta(L_b/\tau)$ constraints, leading to a substantial computational slowdown. In practice, although solving the LP might be faster than its theoretical complexity suggests (e.g., when using the practical simplex method), it remains markedly slower than our combinatorial approach due to its huge size.

\section{Simulations of Algorithms}
\label{sec:Simulations}

In this section, we simulate both the binary search algorithm and the LP algorithm to demonstrate the structures of the corresponding optimal solutions under different setups. We also give practical evidence that the optimal solution may not be unique, and some point in the optimal distribution can be neither $CR$ nor mass tight, which is different from the optimal structure in the classical setting with tail bounds. We note that although some computational error is inevitable by tuning the guess of OPT and scaling the problem when using the binary search method, it is much faster to get an arbitrarily close answer than solving the LP with the same problem scale. 

In Figure~\ref{fig:two different solutions with the same input}, we set~$a=0.8,\delta=0.05$ and~$\gamma=2-a$. We show figures of the optimal purchase distribution, the expected $CR$ (showing $CR$ constraints) and the probability of exceeding~$\gamma$ (total mass in the bad interval for each point, showing mass constraints) given by both the LP algorithm (Figure~\ref{fig:world2 LP}) and the binary search algorithm (Figure~\ref{fig: world 2 greedy with finite support}). Note that both algorithms give the same optimal $CR$ value (a bit larger than 1.1) to our tail bound problem. By our criterion in Section~\ref{subsec: general structure of optimal solutions}, we know that the optimal solution must be in world~2. In other words, this is a solution in the small $\delta$ regime as we discussed earlier. Since our subroutine in the binary search algorithm follows a greedy approach, the structure of the optimal solution it produces is expected to match the form described in the full structural Theorem~\ref{our contribution:thm 4, greedy solution, small delta} before $L_b$. As Figure~\ref{fig: world 2 greedy with finite support} shows, it confirms that the solution returned by the binary search algorithm exhibits this structure: Before time 1, either $CR$ or mass constraint is tight. After time 1 (rescaled by $\tau=0.001$), $CR$ keeps tight until mass becomes tight at some point until $L_b = \frac{\gamma-1}{1-a\gamma}=5$. Then, by placing the remaining mass passed $L_b$ to $\infty$, mass is tight at $\infty$. Finally, we output the corresponding distribution and plot the figures. Figure~\ref{fig:exponential_interrupted_by_zero_intervals} serves as a supplementary illustration, depicting partially the mass distribution we give in Figure~\ref{fig: world 2 greedy with finite support} under the same setup. It shows that we have an exponentially increasing distribution to make $CR$ tight after time 1 at the beginning, then when it becomes mass tight, we have exponential functions (with repeating range of values) interrupted by a bunch of zero intervals. In comparison, the optimal distribution given by the LP algorithm in Figure~\ref{fig:world2 LP} is different from the binary search solution. We show its behavior from 0 to $L_b$. This simulation indicates that the optimal solution is not unique in our small $\delta$ regime. As we promised, one can see that neither $CR$ nor mass constraint is tight at some point in the given solution. We give more plots (Figure~\ref{fig:no fixed pattern for LP output}) showing that solutions returned by the LP algorithms do not have fixed structures. The distribution given by the LP algorithm can be non-greedy in $[1,L_b]$, thus the optimal purchase distribution is not unique even before $L_b$.

We also give the plot when the problem is in the large $\delta$ regime, by setting~$a=0.5, \delta=0.25$ and $\gamma=2-a$. Both the binary search algorithm and the LP algorithm give the same Figure~\ref{fig: unique solution in world 1}. This aligns with the structural property when the optimal solution is in world~1 since $\frac{OPT-1}{\frac{1}{a}-1} < \delta$ in our setup. Again, before time 1, either $CR$ or mass is tight. After time 1, we have a zero suffix which makes $CR$ always tight and guarantees the uniqueness.

\begin{figure}[H]
    \centering
    \begin{subfigure}[b]{0.45\textwidth}
        \centering
        \includegraphics[width=\textwidth]{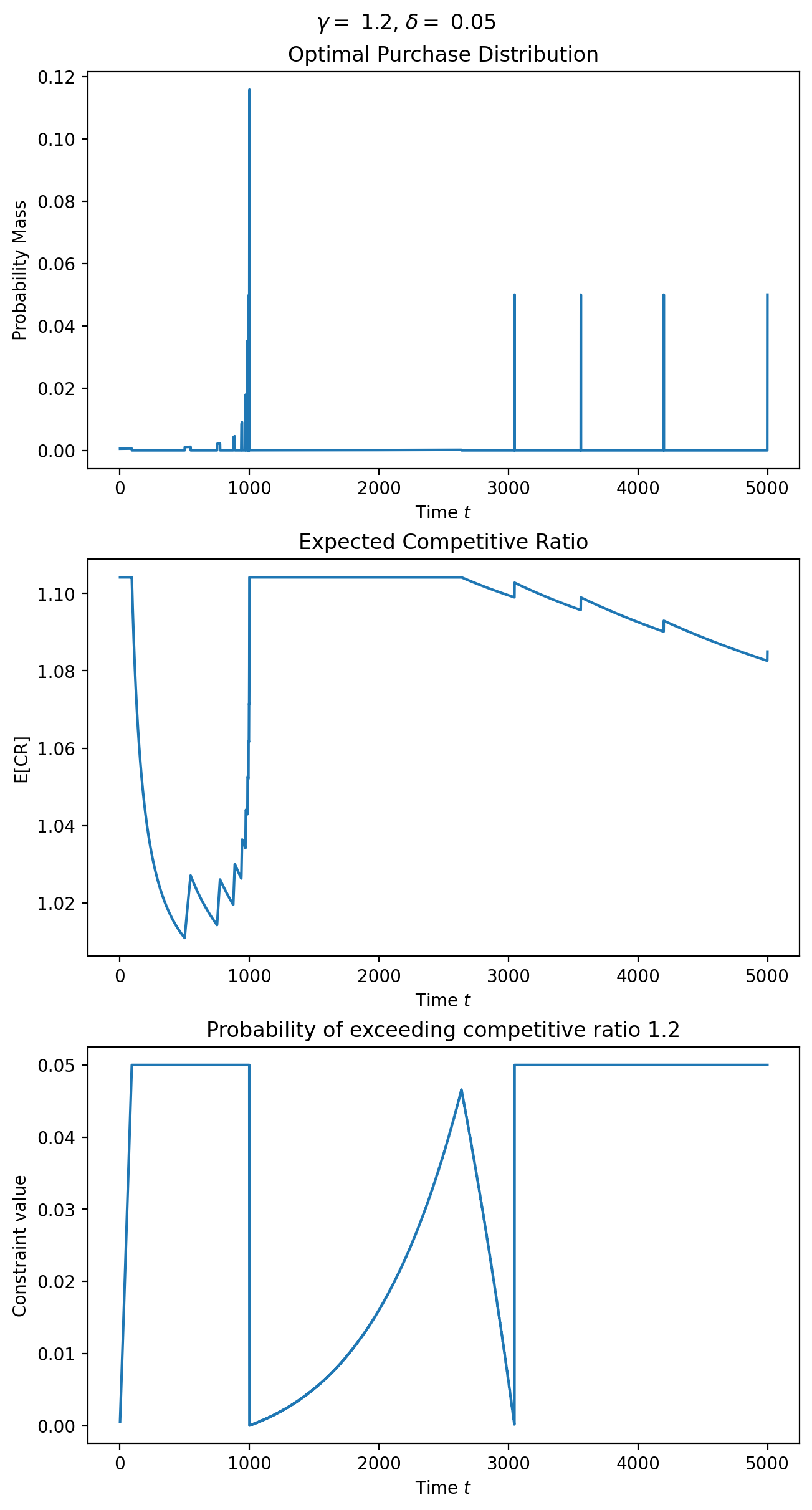}
        \caption{An optimal solution with a finite support $[0,L_b] \cup \{\infty\}$ for $a=0.8, \delta=0.05$ and $\gamma=2-a$ returned by the Linear Program.}
        \label{fig:world2 LP}
    \end{subfigure}
    \hfill
    \begin{subfigure}[b]{0.45\textwidth}
        \centering
        \includegraphics[width=\textwidth,height=13cm]{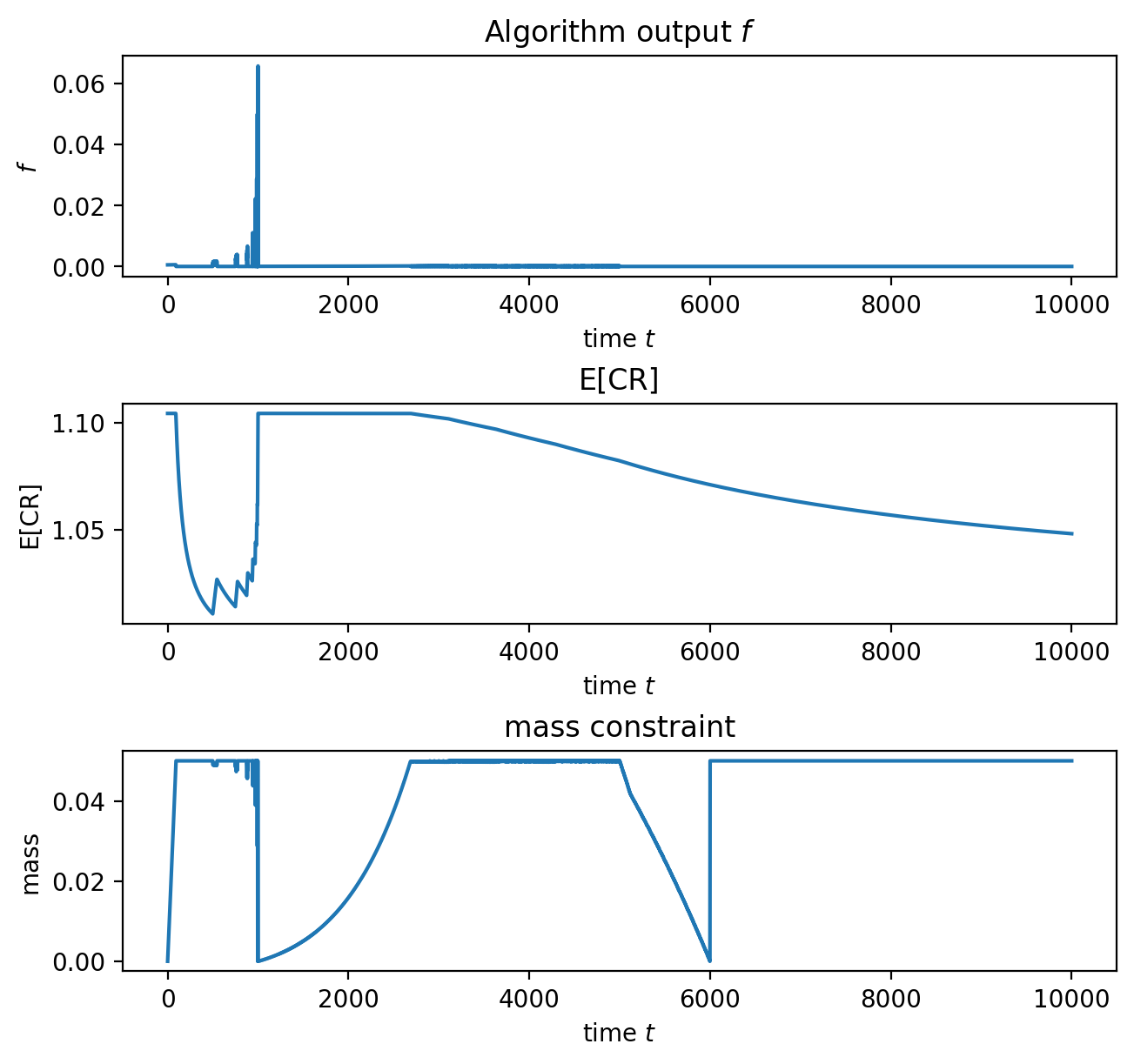}
        \caption{The greedy solution for $a=0.8, \delta=0.05$ and $\gamma=2-a$ returned by the  binary search algorithm. For the output distribution, we plot the distribution by putting all mass after $L_b$ to $\infty$, to give the solution matching the form of the LP solution.}
        \label{fig: world 2 greedy with finite support}
    \end{subfigure}
    \caption{Two optimal solutions for $a=0.8, \delta=0.05$ and $\gamma=2-a$. On each side, the above figure is the mass distribution. The middle is the expected $CR$ constraint as a function of the adversary's choice $t$. The bottom is the probability of exceeding $\gamma$ as a function of the adversary's choice $t$, that is, the total mass in the bad interval of each $t$.}
    \label{fig:two different solutions with the same input}
\end{figure}

\begin{figure}[H]
    \centering
    \begin{subfigure}[b]{0.45\textwidth}
        \centering
        \includegraphics[width=\textwidth]{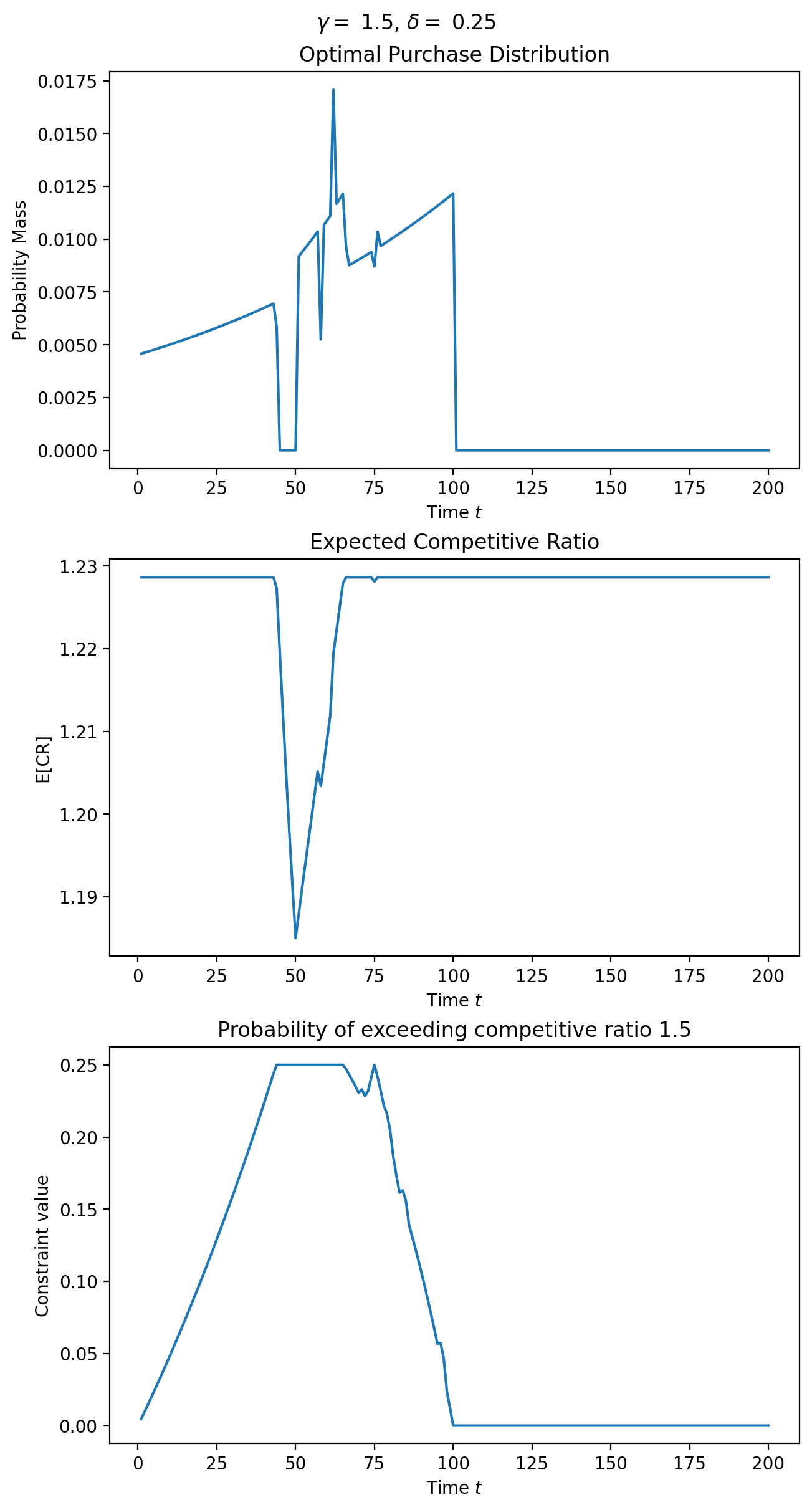}
        \caption{The optimal distribution for $a=0.5, \delta=0.25$ and $\gamma=2-a$. Note that for this input, after time 1 (100 with scaling) $CR$ keeps tight. Our problem is in a large $\delta$ regime.}
        \label{fig: unique solution in world 1}
    \end{subfigure}
    \hfill
    \begin{subfigure}[b]{0.45\textwidth}
        \centering
        \includegraphics[width=\textwidth,height=13.5cm]{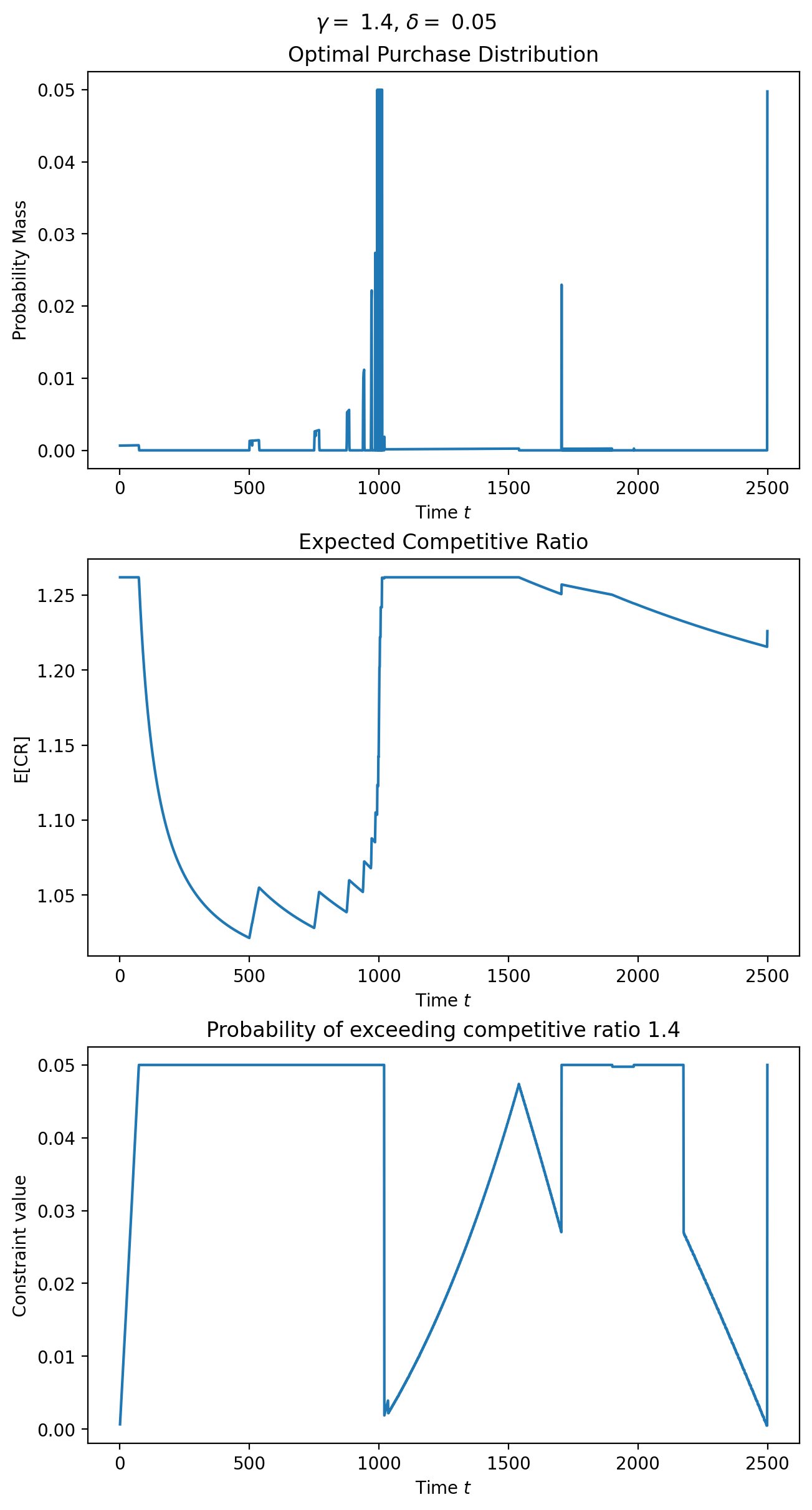}
        \caption{An optimal solution with a finite support for $a=0.6, \delta=0.05$ and $\gamma=2-a$ returned by the Linear Program. It shows that the output of the LP may have different structures (patterns).}
        \label{fig:no fixed pattern for LP output}
    \end{subfigure}
    \caption{On each side, the meaning of each figure is the same as Figure~\ref{fig:two different solutions with the same input}.}
    \label{fig:more experiments}
\end{figure}

\begin{figure}[H]
    \centering
    \includegraphics[width=15cm, height=15cm]{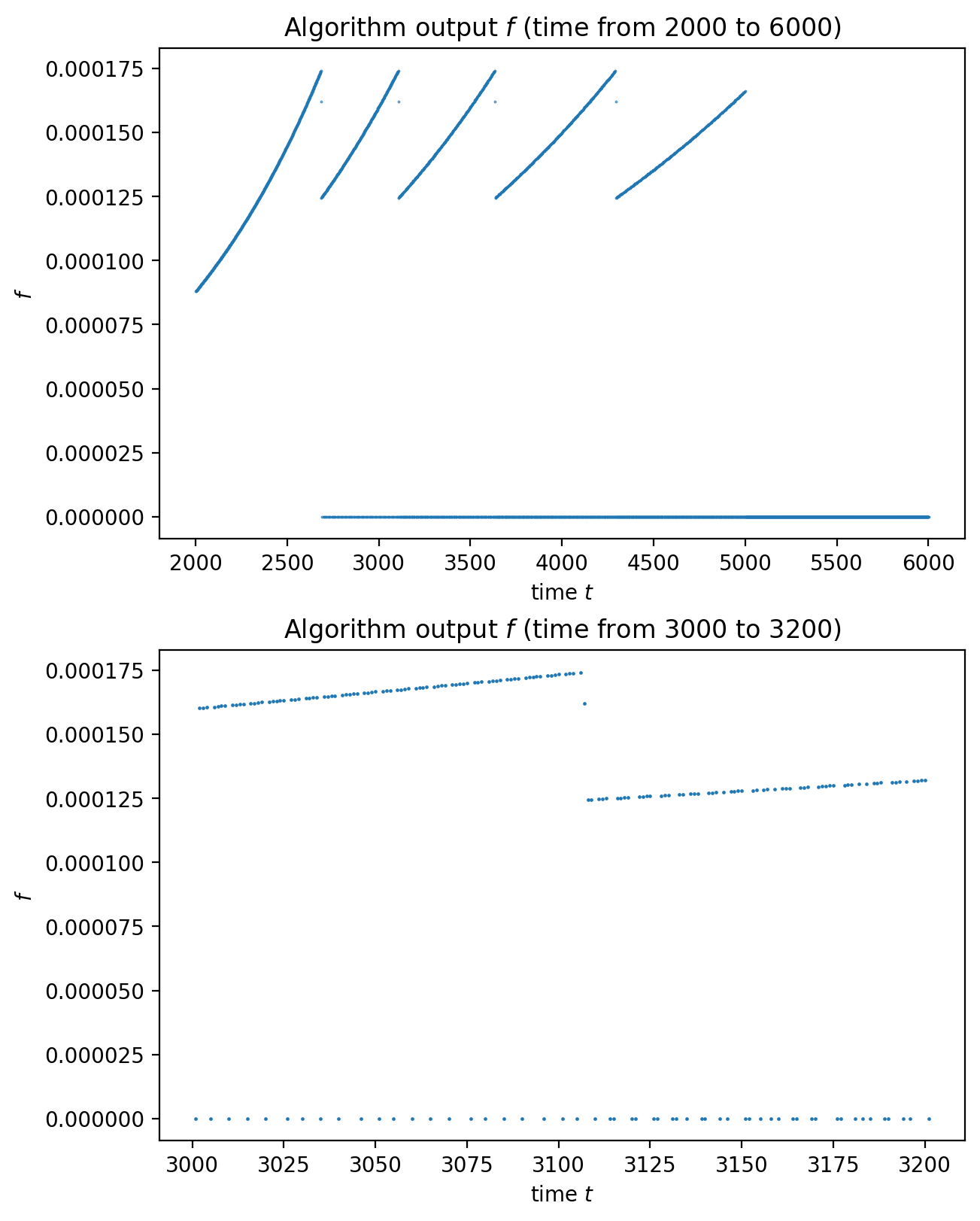}
    \caption{A zoomed-in of the mass distribution between time 2 and 6 of the binary search solution for $a=0.8, \delta=0.05$ and $\gamma=2-a$ (Figure~\ref{fig: world 2 greedy with finite support}) by placing the mass after $L_b$ to $\infty$. A demonstration of how does the solution look like when mass constraint keeps tight, i.e. how exponential functions are interrupted by a bunch of zero intervals. The top figure displays the mass distribution between time 2 and 6 for the optimal solution. The bottom figure, as a further zoomed-in version from time 3 to 3.2, shows that the distribution is an alternation between zeros and exponentials, not overlapping.}
    \label{fig:exponential_interrupted_by_zero_intervals}
\end{figure}

\end{document}